\documentclass[12pt]{article}
\usepackage{amsmath}
\usepackage{graphicx,psfrag,epsf}
\usepackage{enumerate}
\usepackage{natbib}
\usepackage{url} 
\usepackage{longtable}
\usepackage{hyperref}       
\usepackage{url}            
\usepackage{booktabs}       
\usepackage{amsfonts,amsmath,subcaption}       
\usepackage{enumitem,bbm,graphicx}
 \usepackage{nicefrac}       
\usepackage{microtype}  
\newcommand{\blind}{1}

\usepackage{YK}

\def\halpha{\widehat{\alpha}}
\def\heta{\widehat{\eta}}
\def\hB{\widehat{B}}
\def\logit{\operatorname{logit}}
\def\reals{\mathbf{R}}
\def\sigmoid{\operatorname{sigmoid}}
\def\logit{\text{logit}}

\usepackage{array}
\newcounter{rowcount}
\begin{document}

\def\spacingset#1{\renewcommand{\baselinestretch}%
{#1}\small\normalsize} \spacingset{1}


\if1\blind
{
  \title{\bf A linear adjustment based approach to posterior drift in transfer learning}
  \author{Subha Maity\textsuperscript{1}, 
    Diptavo Dutta\textsuperscript{2}, Jonathan Terhorst\textsuperscript{1},\\
    Yuekai Sun\textsuperscript{1} and Moulinath Banerjee\textsuperscript{1}.\\
    \vspace{0.2cm}\\
    \textsuperscript{1}Department of Statistics,
    University of Michigan,\\
    \textsuperscript{2}Department of Biostatistics and Department of Biomedical Engineering,\\
    Johns Hopkins University.
    }
    
    \date{}
  \maketitle
} \fi

\if0\blind
{
  \bigskip
  \bigskip
  \bigskip
  \begin{center}
    {\LARGE\bf A linear adjustment based approach to posterior drift in transfer learning}
\end{center}
  \medskip
} \fi

\bigskip
\begin{abstract}
We present a new model and methods for the posterior drift problem where the regression function 
in the target domain is modeled as a linear adjustment (on an appropriate scale) of that in the source domain, an idea 
that inherits the simplicity and the usefulness of generalized linear models and accelerated failure time models from the classical statistics literature, and study the theoretical properties of our proposed estimator in the binary classification problem. Our approach is shown to be flexible and applicable in a variety of statistical settings, and can be adopted to transfer learning problems in various domains including epidemiology, genetics and biomedicine. As a concrete application, we illustrate the power of our approach through mortality prediction for British Asians by borrowing strength from similar data from the larger pool of British Caucasians, using the UK Biobank data. 
\end{abstract}

\noindent%
\vfill

\newpage
\spacingset{1.45} 
\section{Introduction}
Transfer learning has emerged as a popular and active area of research in the field of statistical learning \citep{torrey2010transfer,pan2009survey,weiss2016survey,zhuang2020comprehensive}. In a broad sense transfer learning arises when we have an abundance of data from a `source domain' which is not of primary interest, and a relatively smaller amount of data from a `target domain', of primary interest, and would like to inform our analysis of the target domain by borrowing strength from the data-rich source domain. 
The key aspect of transfer learning lies in the fact that the source and target data come from different distributions, and therefore, classifiers/estimators trained exclusively on the source may exhibit poor performance in the target domain. There is a wide range of works on how one may use source knowledge in conjunction with limited training data from the target to improve test performance in the latter domain, which, typically, impose some structural similarity between the source and target domains and design algorithms that exploit these similarities to adapt classifiers trained on the source to the target domain \citep{lipton2018Detecting,combes2020Domain}. Complementary to this methodological line of work, there is a more theoretical line of work that studies the fundamental limits of knowledge transfer under certain structural similarities between the source and target domain \citep{cai2019Transfer,kpotufe2018Marginal,maity2020Minimax,reeve2021adaptive}. 

In this paper we specifically consider the problem known as \emph{posterior drift} \citep{cai2019Transfer,liu2020computationally} under a prediction task. The setup assumes that the Bayes classifiers for the source and target populations are different. We primarily restrict ourselves to the binary non-parametric classification setup. So, it is assumed that the posterior probability or the conditional probability of success given covariates are different for the source and the target distributions.  Let $\cX$ be the feature space. Both source ($\bbP$) and target ($\bbQ$) distributions are probability distribution on the space $\cX \times \{0, 1\}$. The conditional probabilities, also known as regression functions, for source and target are denoted as $\eta_P$ and $\eta_Q$ respectively, \ie 
\[ \eta_P(x) = \bbP(y = 1|x), \quad \eta_Q(x) = \bbQ(y = 1|x),  \]
and $\eta_P \neq \eta_Q$.
This type of setup is of particular interest when the experimental setup changes inducing differences in the outcome generating  mechanism. \cite{cai2019Transfer},  in their pioneering work, studied the fundamental limits of knowledge transfer in posterior drift subject to the additional assumption that the decision boundaries in the source and target domains are identical. In this paper, we consider a more general setting in which the decision boundaries are permitted to differ, but the difference is restricted parametrically. In particular, we assume the source and target regression functions are related via
\begin{equation}
\gamma(\eta_Q(x)) = \gamma(\eta_P(x)) + \beta^\top T(x),~~~  \beta\in \reals^{d_T}, T(x) \in \reals^{d_T},
\label{eq:posterior-drift-model}
\end{equation}
where $\gamma$ is a link function and $T(x)$ is a low-dimensional transformation of the covariates. The term $\beta^\top T(x)$ captures the difference between the regression function in the source and target domains. Intuitively, the unknown transfer parameter $\beta$ is much easier to learn than the (potentially complex) regression function $\eta_P$, so transfer is possible even with few samples from the target domain. We note that a special case of this model is label shift \citep{iyer2014Maximum,azizzadenesheli2019Regularized, lipton2018Detecting}: $\gamma$ is the logit link function, $T(x) = 1$, and $\beta = \log\frac{\bbQ(Y=0)}{\bbQ(Y=1)} - \log\frac{\bbP(Y=0)}{\bbP(Y=1)}$.

This linear shift model is motivated from the celebrated Cox proportional hazard model or Cox regression model \cite{cox1972regression}. In this paper, the time dependent integrated hazard is modeled as an exponential of linear shift (of covariates) over the time dependent baseline:
$\Lambda(t|x) = \Lambda_0(t) \exp(\beta^\top x)$. 
Recalling the definition of integrated hazard $\Lambda(t) = - \log (1 - F(t))$ we see that the Cox-PH model has the following equivalent expression.
\[
\log \circ (-\log) \big(1 - F(t|x)\big) = \log \circ (-\log) \big(1 - F_0(t)\big) + \beta^\top x
\] Here,  the transformed conditional survival function $1 - F(\cdot | x)$ is obtained from the transformed baseline survival function $1-F_0$ and a linear shift over the covariates, and the link function $\gamma$ is $\log \circ (-\log)$.

In our paper we mainly focus on the logistic link function, which is of particular interest for classification tasks, as it is easy to work with and is amenable to interpretation. While our results can be generalized without much difficulty to other (known) links, we note that mis-specification of the link function is not a major concern under the prediction task. But as far as inference tasks for the linear shift parameters is concerned, it is important that the link function be prescribed correctly, as otherwise the estimates of the linear shift parameters could be biased. This naturally leads to the question of estimating the link function in such situations. The concluding section provides some discussion for the case where $\gamma$ is unknown.

We study the convergence rate of the target excess risk \eqref{eq:target-excess-risk} in the linear shift model with a logistic link function and  present matching upper and lower bounds subject to standard smoothness assumptions on the regression functions and a margin condition on the noise level (\cf\  Theorems \ref{thm:upper-bound} and  \ref{th:lower-bound}). We also demonstrate the practical relevance of the linear shift model with a mortality prediction task with the UK Bio-bank dataset \citep{biobank2014uk}. Here, the transfer learning task is adapting a classifier trained on individuals with European ancestry to individuals with Asian and Asian British ancestry. This task is especially suitable for transfer learning because most of the individuals in the UK Bio-bank dataset have European ancestry, whilst only a small fraction have Asian and Asian British ancestry. Thus, it is necessary to borrow strength from the individuals with European ancestry to improve the prediction performance on Asian and Asian British individuals. Recently several machine learning approaches including random forests, neural networks and others have been developed to improve the accuracy of mortality prediction models. Here, we  used two approaches: (a) logistic regression and (b) random forest, to predict the 10-year overall (all cause) mortality in individuals of Asian ancestry. For both the classifiers, we demonstrate the improved accuracy of our approach over standard baseline models. The rest of this paper is organized as follows:

In Section \ref{sec:posterior-drift} we 
investigated the fundamental difficulty of prediction in the posterior drift problem and developed a linear adjustment method that achieves this limit. In Section \ref{sec:simulation} we present a performative study of our model in a synthetic data. In Section \ref{sec:ukbb} we show the benefit of linear shift model for predicting mortality in UK Biobank data. Finally, in Section \ref{sec:dsicussion} we discuss several  possible avenues for extending the linearly shifted posterior drift model.

\section{Posterior drift for logistic model}

\label{sec:posterior-drift}

We start by introducing some notations. We  denote the entropy loss function as $\ell: \{0, 1\} \times \reals \to \reals, \ \ell(y, a) = - ya + \log\big(1+e^a\big)$, $\gamma$   denotes the link function for logistic regression, \ie, $\gamma(x) = \log(x/(1-x))$ and $\sigma$  denotes the corresponding inverse link function, \ie, $\sigma(x) = \gamma^{-1}(x) = 1/(1+e^{-x})$. Throughout our study $m$  denotes the source sample size and $n$ the  target sample size and we assume $m > n$. We denote the   Lebesgue measure on a Euclidean space by $\lambda$ and a ball centered at $x$ with radius $r$ by $B(x, r) \triangleq \{y: \|y-x\|_2\le r\}$. 
For a vector $x = (x_1, \dots, x_d)\in \reals^d$ and a scalar $c\in \reals$ the quantity $x + c$ denotes the operation $(x_1 + c,\dots, x_d + c)$. 

In this section we discuss the posterior drift adaptation to the target distribution under the logistic regression model. The adaptation method is agnostic to modeling assumptions for the source regression function $\eta_P.$ Recall the posterior drift model \eqref{eq:posterior-drift-model}. \[
\gamma (\eta_Q(x)) = \gamma (\eta_P(x)) + \beta^\top T(x)\,.
\] We call the parameters $\beta$ as  transfer parameters. 

For the moment, assume that there exists a suitable estimate of  $\eta_P$ denoted as $\hat\eta_P(x).$ We provide details about  possible model choices for $\eta_P$ later. We plug in this estimate into the non-parametric component of  $\gamma(\eta_Q)$ 
and estimate the transfer parameters by fitting logistic regression to the target data: 
\begin{equation}
  \hat \beta \in  \argmin_{\beta}  \frac1n \sum_{i=1}^n \ell\Big(y_{i}^{(Q)}, \gamma(\hat\eta_P(x_i^{(Q)})) +  \beta^\top T\big(x^{(Q)}_{i}\big) \Big)\,.
  \label{eq:offset-estimate}
\end{equation}
Note that $\hat \beta$ is just a maximum likelihood estimator.

Having obtained the parametric estimates from  \eqref{eq:offset-estimate}, we
estimate the target regression function as
\begin{equation}
    \label{eq:target-reg-fn-estimate}
    \hat \eta_Q(x) = \sigma \Big( \gamma(\hat\eta_P(x))+   \hat \beta ^\top T(x)\Big), 
\end{equation} and 
predict the target label of an observation with covariate $x$ as 
\begin{equation}
\hat f (x) = \bbI \Big\{ \hat \eta_Q(x) > \frac 12 \Big\}.
    \label{eq:target-classifier}
\end{equation}

\subsection{Convergence of the excess risk}
We start by studying the upper bound of the excess risk for the classifier \eqref{eq:target-classifier}. We recall that the excess risk of a classifier $f$ is the difference in the mis-classfication error compared to the Bayes classifier, which has the minimum mis-classification error. 
\begin{equation}
\cE_{\bbQ}(f) = \bbQ(f(x) \neq y) - \bbQ(f_*(x) \neq y), ~~~ f_* \text{ is the Bayes classifier for }\bbQ.
\label{eq:target-excess-risk}
\end{equation}

To perform a non-asymptotic study on the  behavior of the target excess risk we require some structural assumptions on the source and target distributions.  The first assumption ensures that $\hat \eta_P$ has sub-Gaussian concentration around $\eta_P$. 

\begin{assumption}[Concentration of the estimated source regression function]
There exists an estimator $\hat \eta_P$ of the function $\eta_P$ such that for some $c_1, c_2, >0$ for almost surely $[P_X]$ any $x$ we have 
\[
\bbP \Big(\big|\hat \eta_P(x) - \eta_P(x)\big| > t\Big) \le c_1 \exp \left( - c_2 \left(\frac{t}{r_m}\right)^2 \right), \ t > 0.
\]
\label{assumption:source-reg-fn-concentration}
\end{assumption} 

Note that the rate $r_m$ is dependent on the source sample size $m$ and decreases to 0 as $m \to \infty.$  This  concentration rate  appears in the convergence rate for the excess risk of the target classifier and dictates the contribution from source data. A higher source sample size results in a  smaller $r_m$ which  leads to a better target classifier. 

We next consider a technical assumption on the target marginal distribution of $x$. If we have a point $x$ which is not in the support (see Definition \ref{def:support}) of $\bbP_X$ but belongs to the support of $\bbQ_X$, \ie,  $\bbP_X(B(x, r)) = 0$ and $\bbQ\big(B(x, r)\big) > 0$ for some $r>0$, then $\eta_P(x)$ is not estimable from the source data, and the plug-in approach fails at the point $x$. Thus, we assume: 
\begin{assumption}
\label{assump:dominated-measure} The support of $\bbQ_X$ is included in the support of $\bbP_X$. 
\end{assumption}

The next assumption is standard in the literature on maximum likelihood estimation. We require the expected loss function to be strongly convex. 
\begin{assumption}[Strong convexity]
Define $m(\beta) = \bbE_{\bbQ} \Big[ \ell \big(y, \gamma(\eta_P(x)\big) + \beta^\top T(x))\Big].$ Assume that  $m(\beta)$ is strongly convex with  a constant  $\kappa > 0$, \ie, for any $\beta_1, \beta_2$ we have  \[
m(\beta_2)\ge m(\beta_1) + \nabla_\beta m(\beta_1)^\top (\beta_2 - \beta_1) +  \kappa \|\beta_1 -\beta_2\|_2^2\,.
\] 
\label{assump:strong-convexity}
\end{assumption} 

The strong convexity for expected log-likelihood enables us to establish consistency of maximum likelihood estimate. If $\beta_0$ is the true parameter value then $\nabla_\beta m(\beta_0) = 0$. Under strong convexity, $m(\beta) \ge m(\beta_0) + \kappa \|\beta - \beta_0\|_2^2$ and  $m(\beta)$ is uniquely at minimized at $\beta_0$.

In order to establish a finite sample bound on the excess risk, we need some control on the tail-behavior of $T(x)$: 

\begin{assumption}[Sub-Gaussian tail assumption]
For $(x, y)\sim \bbQ$, the random variable $T(x)$ has  sub-Gaussian tail, \ie, there exists an $M>0$ such that for any $t > 0$
\[
\sup_{\|a\|_2 = 1}\bbQ \bigg(a ^\top T(x)  > t \bigg) \le \exp\left({- \frac{t^2}{M}}\right)\,.
\]
\label{ass:sub-gaussian}
\end{assumption} 

The above assumptions are enough to establish a concentration bound for $\hat\beta$ and $\hat \eta_Q(x)$. However, in order to translate these bounds to one on the excess risk, we need a final ingredient in our list of assumptions, namely a margin condition on $\eta_Q$. 
\begin{assumption}[Margin condition]
\label{assump:margin-condition}
There exist constants $\alpha, C_\alpha>0$ such that for any $t>0$ we have the following probability bound. 
\[
\bbQ_X\left(0<\Big|\eta_Q(x) - \frac12 \Big|<t\right) \le C_\alpha t^\alpha\,. 
\]
\end{assumption}

The margin condition appears frequently in the literature on  non-parametric classification \citep{audibert2007fast,tsybakov2009Introduction,cai2019Transfer,kpotufe2018Marginal,maity2020Minimax} and restricts the probability  of the region in which $\eta_Q(x)$ is close to $1/2$. The key implication of the margin condition is that we can control  the probability that $\hat \eta_Q(x)$ and $\eta_Q(x)$ lie on opposite sides of $1/2$, thereby  reducing the chance of disagreement between $\hat f$ and Bayes classifier. Indeed, a higher value of $\alpha$ in margin condition leads to a faster rate for excess risk, as seen in Theorem \ref{thm:upper-bound}, and several other papers \citep{audibert2007fast, cai2019Transfer, kpotufe2018Marginal, maity2020Minimax}.

Under the above assumptions we establish a non-asymptotic rate of convergence for the excess risk. We define the class $\cP \equiv \cP\big(T, \{r_m\}_{m\ge 1}, \alpha \big)$ as the set of probability pairs $(\bbP, \bbQ)$ satisfying the condition \ref{eq:posterior-drift-model} and assumptions \ref{assumption:source-reg-fn-concentration}, \ref{assump:dominated-measure}, \ref{assump:strong-convexity}, \ref{ass:sub-gaussian} and \ref{assump:margin-condition}. Note that this class depends on the constants that appear in the above assumptions but we will not explicitly mention them subsequently, for notational convenience. In our upper bound study the rate is presented only in terms of the quantities $m, n$ and $\alpha$ but the constant that appears in the upper bound depends on the constants in the assumptions. 

\begin{theorem}[Upper bound]
\label{thm:upper-bound}

Let $H = \bbP^{\otimes m} \otimes \bbQ^{\otimes n}$. There is a constant $C>0$ which is independent of $m, n$ and $\alpha$ such that the following holds. 
\[
\sup_{(\bbP, \bbQ) \in \cP} \bE_{H}\big[\cE_\bbQ(\hat f)\big] \le C\left(r_m\sqrt{\log( n)} + \sqrt{\frac {d_T}n}\right)^{1+\alpha}\,. 
\]

\end{theorem}

A detailed proof of the upper bound is provided in Appendix \ref{proof:ub} but here is the basic idea: We first establish that $\hat \beta$ exhibits sub-Gaussian concentration around  $\beta$. Next, we combine the concentration of $\hat \beta$ and $\hat \eta_P$ (Assumption \ref{assumption:source-reg-fn-concentration}) to establish a concentration bound for $\hat \eta_Q$. Finally  the concentration on $\hat \eta_Q$ in conjunction with the margin condition (Assumption \ref{assump:margin-condition}) delivers the non-asymptotic upper bound on excess risk. 

The upper bound quantifies the non-asymptotic behavior of the expected excess risk in terms of the quantities $r_m, n$ and $\alpha$. 

\begin{remark}
The extra $\sqrt{\log(n)}$ term, which does not manifest in the lower bound as we will see later, appears to be unavoidable for providing a concentration bound on the quantity $$E_{\bbQ_X}\big[(\hat \eta_P (x) - \eta_P(x))^2\big] \equiv \int\,(\hat \eta_P (x) - \eta_P(x))^2\,d\bbQ_X(x) \,.$$ 
The event $\big\{|\hat \eta_P(x) - \eta_P(x)|> t\big\}$ that appears in the display of Assumption \ref{assumption:source-reg-fn-concentration} typically depends on $x$. In order to tackle the above expectation, we need to take a union bound of such events over the $x$ realizations in target data, which produces the $\sqrt{\log n}$ term via a supremum argument. 
However, it does not appear under linear modelling assumptions on  $\eta_P(x)$, e.g. in terms of main effects and interactions as used for our simulations and data-analysis. If  $\gamma\big(\eta_P(x)\big) = \theta^\top S(x)$ for some $S(x), \theta \in \reals^{d_S}$, then we note that: 
\[
\begin{aligned}
E_{\bbQ_X}\big[\big(\hat \eta_P (x) - \eta_P(x))^2\big] & = E_{\bbQ_X}\big[\big(\sigma(\hat \theta^\top S(x)) - \sigma( \theta^\top S(x))\big)^2\big]\\
& \le \frac14 E_{\bbQ_X}\big[\big( (\hat \theta - \theta )^\top S(x)\big)^2\big] \le \frac{\|\hat \theta - \theta\|_2^2E_{\bbQ_X}\big[\|S(x)\|_2^2\big]}{4} \,.
\end{aligned}
\]
If  $\hat\theta$ concentrates  around $\theta$ at a rate $r_m$  (in the usual case $r_m = \sqrt{d_S/m}$) and $E_{\bbQ_X}\big[\|S(x)\|_2^2\big] < \infty$ then  $E_{\bbQ_X}\big[\big(\hat \eta_P (x) - \eta_P(x))^2\big]$ also concentrates around zero at a rate $r_m^2$ and the $\sqrt{\log(n)}$ term does not appear. 
\end{remark}

Another interesting question regarding the convergence of excess risk is whether our upper bound is rate optimal. We address this question, partly, under non-parametric modelling of the source regression function $\eta_P.$ There is a vast literature on non-parametric estimation of the source regression function $\eta_P$,  typically restricted to the case 
where $\cX$ is a compact subset of $\reals^{d_F}$. Without loss of generality we will assume $\cX\subset [0, 1]^d$. We will assume the following two conditions on $\bbP_X$: 
The first condition assumes regular support and strong density for $\bbP_X$ \cite[Definition 2.2]{audibert2007fast} \footnote{A lower bound result could also have been proved under a {mild density assumption} \cite[Definition 2.1]{audibert2007fast} at the expense of a slower convergence rate.}.
\begin{assumption}[Regular support and strong density]
\label{assump:strong-density}
The followings are true for $\bbP_X$. 
\begin{enumerate}
    \item \emph{Regular support}: The set $S_P \triangleq \text{ support of }\bbP_X$ is compact and  there exist $c_0, r_0>0$ such that  $S_P$ is $(c_0, r_0)$-regular, \ie, \ \[
    \lambda \big[B(x, r)\cap S_P\big] \ge c_0\lambda\big[B(x, r)\big], \text{ for any } 0 < r \le r_0. 
    \]
    \item \emph{Strong density}: There exist $0< \mu_{\text{min}} < \mu_{\text{max}}<\infty  $ such that the density $p_X$ of $\bbP_X$ satisfies \[
    \mu_{\text{min}} \le p_X(x)\le  \mu_{\text{max}},\  x \in S_P\,. 
    \]
\end{enumerate}
\end{assumption}
 
 The next condition pertains to the smoothness of $\eta_P$. For a function $g : \reals^{d_F} \to \reals$ and a non-negative integer $b$ we denote $g_x^{(b)}(\cdot)$ as the $b$-th order Taylor polynomial of the function $g$ (see Definition \ref{def:taylor-poly}). With this notation the smoothness condition can be stated thus: 
\begin{assumption}[Smooth regression function]
\label{assump:smoothness}
There exists some positive constants $\beta, C_\beta>0$ such that $\eta_P$ is $\lfloor \beta \rfloor$ times continuously differentiable and for any $x, x'\in \cX$ \[
\left|\eta_P(x') - \big(\eta_P\big)_x^{(\lfloor\beta\rfloor )}(x')\right|\le C_\beta \|x-x'\|_2^\beta \,. 
\]
\end{assumption}

In non-parametric classification, the Assumptions \ref{assump:strong-density} and \ref{assump:smoothness} imply tat there exists an $\hat \eta_P$ which satisfies the Assumption \ref{assumption:source-reg-fn-concentration} with $r_m = m^{-\frac{\beta}{2\beta + d_F}}$ (\cite[Theorem 3.2]{audibert2007fast}). For our lower bound calculations we will  consider a new  class of probability pairs  by replacing Assumption \ref{assumption:source-reg-fn-concentration} with the Assumptions \ref{assump:strong-density} and \ref{assump:smoothness} while holding all other assumptions fixed and denote this class as $\cP'$. Note that $\cP' \subset \cP\big(T, \big\{m^{-\frac{\beta}{2\beta + d_F}}\big\}_{m\ge 1}, \alpha\big)$. 
For the following minimax result, we consider the following hypothesis class: 
\[
\cF_{m, n} = \left\{f: \big(\cX\times \{0, 1\}\big)^{m+n} \to \cC\right\}\,,
\]
where $\cC$ denotes all measurable functions from $g:\cX\to \{0, 1\}$. 
 
 \begin{theorem}[Lower bound] 
 \label{thm:lower-bound}
Let $H = \bbP^{\otimes m} \otimes \bbQ^{\otimes n}$. There exists a constant $c>0$ independent of $m, n, \alpha$  and $ \beta$ such that 
\[ \inf_f\sup_{(\bbP, \bbQ)\in \cP'} \bE_{D\sim H} \big[\cE_{\bbQ}\big(f(D)\big)\big] \ge c \left( m^{-\frac\beta{2\beta+d_F}} + \sqrt{\frac {d_T} {n}} \right)^{1 + \alpha}\,.  \]
\label{th:lower-bound}
\end{theorem}

We provide a proof of the above theorem in Section \ref{sec:proof-lb}.

\begin{remark}
Firstly, we notice that the lower bound in Theorem \ref{thm:lower-bound} matches with the upper bound in Theorem \ref{thm:upper-bound} with $r_m = m^{-\frac{\beta}{2\beta + d_F}}$ upto a logarithm term ($\log (n)$) and a constant term which is independent of the parameters $m, n, \alpha$ and $ \beta$. We furthermore notice that the probability class $\cP'$ used to derive the lower bound is a subclass of the probability class $\cP\big(T, \big\{m^{-\frac{\beta}{2\beta + d_F}}\big\}_{m\ge 1}, \alpha\big)$ in the upper bound. With these two observations we conclude that the classifier $\hat f$ is minimax rate optimal (upto a $\log(n)$ term) in  the class $\cP\big(T, \big\{m^{-\frac{\beta}{2\beta + d_F}}\big\}_{m\ge 1}, \alpha\big)$. 
\end{remark}

\begin{remark}
Note that the rate does not adapt to the difference between the source and target distributions, meaning that under the sub-case that $\beta = 0$ the rate does not reduce to $(m+n)^{-\frac{\beta(1+\alpha)}{2\beta + d_F}} \asymp m^{-\frac{\beta(1+\alpha)}{2\beta + d_F}}$, the non-parametric rate of convergence for \iid\ data \citep{audibert2007fast}. In this regard, we claim that it is impossible to obtain an adaptive rate in this situation. To provide insight, we consider the simple case where $\eta_P \equiv 1/2$ and $T(x) = 1$, whence the linear shift model reduces to Bernoulli distributions on source and target with  probabilities  $\frac12$ and $p = \frac{e^\beta}{ 1+ e^\beta}$ for class one, respectively. To make a prediction on the target labels it is necessary to estimate $p$, which can be done at a rate no faster than $\frac1{\sqrt{n}}$, regardless of whether $\beta = 0 \Leftrightarrow p = \frac12$. So, for a fixed $n$, the excess risk does not converge to zero as $m \to \infty$ regardless of whether $p = \frac12$. 

\end{remark}

  \subsection{Alternative estimate for linear model of source regression function}
  
  In this section we present an alternative estimator based on  linear modeling of the source regression function, which is the strategy we employ later for the analysis of the UK Biobank data. 
Let $S(x) \in \reals^{d_S}$ be some transformation of the covariate and assume the following model for the source regression function. 
  \[
  \gamma(\eta_P(x)) = \xi^\top S(x) + \beta_P^\top T(x),\quad \beta_P\in \reals^{d_S}.
  \] Here, $\xi$ is the coefficient for source corresponding to the covariate $S(x).$ The posterior drift condition \eqref{eq:posterior-drift-model} implies the following form for the target regression function:  
  \[
  \gamma(\eta_Q(x)) = \xi^\top S(x) + \beta_Q^\top T(x),\quad \beta_Q \equiv \beta_P + \beta.
  \]
  The parameters are now estimated as follows: 
  \begin{equation}
      (\hat \xi, \hat \beta_P, \hat \beta) = \argmin\ \frac1{m+n} \sum_{i = 1}^{m+n}\ell \Big( y_i, \xi^\top S(x_i) + \big(\beta_P^\top + \bbI_Q(i)\beta  \big) T(x_i) \Big),
  \end{equation}
  where, $\bbI_Q(i)$ is the indicator variable for identifying whether a sample point is from target population. The above optimization produces the $(m+n)^{-1/2}$-rate for $\hat \xi$ compared to an $m^{-1/2}$-rate if only the $P$-data was used. But the performance gain is marginal if $m \gg n. $

\section{Simulations}
\label{sec:simulation}

In this section, we investigate the statistical properties of the linear shift model. In both the source and the target domains, the covariates $x \in \reals^d$ are generated from $N_d(0, 4\bI_d).$ Given covariates $x = (x_1, \dots, x_d)^\top$, the response $y$ is generated from different Bernoulli distributions for source and target domains with the following regression functions:
\[
\begin{aligned}
P(y = 1| x) &=  \sigma\Bigg(\sum_{j=1}^d (-1)^{j-1}(\xi x_j^2 - \delta x_j) \Bigg),\\
Q(y = 1| x) &=  \sigma\Bigg(\sum_{j=1}^d (-1)^{j-1}(\xi x_j^2 + \delta x_j) \Bigg),
\end{aligned}
\] where $\sigma$ is the inverse of an appropriate link function. Note that the parameter $\xi$ controls the strength of non-linearity: $\xi = 0$ leads to a simple linear model. The parameter $\delta$ determines the similarity between the source and target domains: a larger $\delta$ leads to a greater discrepancy between source and target regression functions.

We first fit a logistic regression model on the source with all main and two way interaction effects (including $x_j^2$'s) as the covariates:
\begin{equation}
\begin{aligned}
\heta_P(x) \triangleq \sigmoid(\halpha_P^Tx + \frac12x^T\hB_Px), \\
(\halpha_P,\hB_P) \in \argmin_{\alpha,B}\frac1n\sum_{i=1}^n\ell(y_i^{(P)},\alpha^Tx_i^{(P)} + \frac12(x_i^{(P)})^TB x_i^{(P)}).
\end{aligned}
\label{eq:source.full}
\end{equation}
For $x\in \reals^d$ the total number of covariates in this model is $2d + \frac{d(d-1)}{2}.$ To adapt this model to the target domain, we refit the main effects on samples from the target domain:
\[
\begin{aligned}
\heta_Q(x) \triangleq \sigmoid(\halpha_Q^Tx + \frac12x^T\hB_Px), \\
\halpha_Q \in \argmin_\alpha\frac1n\sum_{i=1}^n\ell(y_i^{(Q)},\alpha^Tx_i^{(Q)} + \frac12(x_i^{(Q)})^T\hB_Px_i^{(Q)}).
\end{aligned}
\]
We call this classifier \emph{transfer}, and we compared it predictive accuracy to several other models:
\begin{itemize}
\item \emph{source.main} which only fits main effects on source data: 
\[
\begin{aligned}
\heta_P(x) \triangleq \sigmoid(\halpha_P^Tx) \\
\halpha_P \in \argmin_\alpha\frac1n\sum_{i=1}^n\ell(y_i^{(P)},\alpha^Tx_i^{(P)}),
\end{aligned}
\]
\item \emph{source.full} which fits main plus interaction effects on the source data \eqref{eq:source.full},
\item \emph{target.main} which only fits main effects on target data:
\[
\begin{aligned}
\heta_Q(x) \triangleq \sigmoid(\halpha_Q^Tx) \\
\halpha_Q \in \argmin_\alpha\frac1n\sum_{i=1}^n\ell(y_i^{(Q)},\alpha^Tx_i^{(Q)}),
\end{aligned}
\]
\item \emph{target.full} which fits main plus interaction effects but estimates them only from target data:
\begin{equation}
\begin{aligned}
\heta_Q(x) \triangleq \sigmoid(\halpha_P^Tx + \frac12x^T\hB_Px) \\
(\halpha_P,\hB_P) \in \argmin_{\alpha,B}\frac1n\sum_{i=1}^n\ell(y_i^{(Q)},\alpha^Tx_i^{(Q)} + \frac12(x_i^{(Q)})^TB x_i^{(P)}),
\end{aligned}
\label{eq:target.full}
\end{equation}
\item \emph{ideal} which is the population version of \emph{target.full}.
\end{itemize}

In figure \ref{fig:sim-logit}, we show the benefits of linearly shifted transfer learning under a well-specified model. In the \emph{upper left} plot, we see that as the source sample size ($m$) increases, the $Q$-accuracy for {transfer} increases. However, there is a gap between the performance of {transfer} and ideal. This is unsurprising  because the target sample size remains fixed ($n = 100$). In the \emph{upper middle} plot, we see that as $n$ increases, the accuracy of both target.full and transfer models approach that of the ideal model, but the accuracy of transfer is closer to that of ideal. In the \emph{upper right} plot, we see that source.full is slightly more accurate than transfer when $\delta$ is small. In this case, there is very little difference between source and target distribution, so there is little need for transfer learning. On the other hand, the accuracy of both source.full and source.main suffers when $\delta$ is large. Finally, in the \emph{lower middle} plot, we see that as the dimension increases, the transfer model is the most accurate (except the ideal model); it is the most statistically efficient because it has the fewest parameters to estimate. 

\begin{figure}
    \centering
    \includegraphics[scale = 0.4]{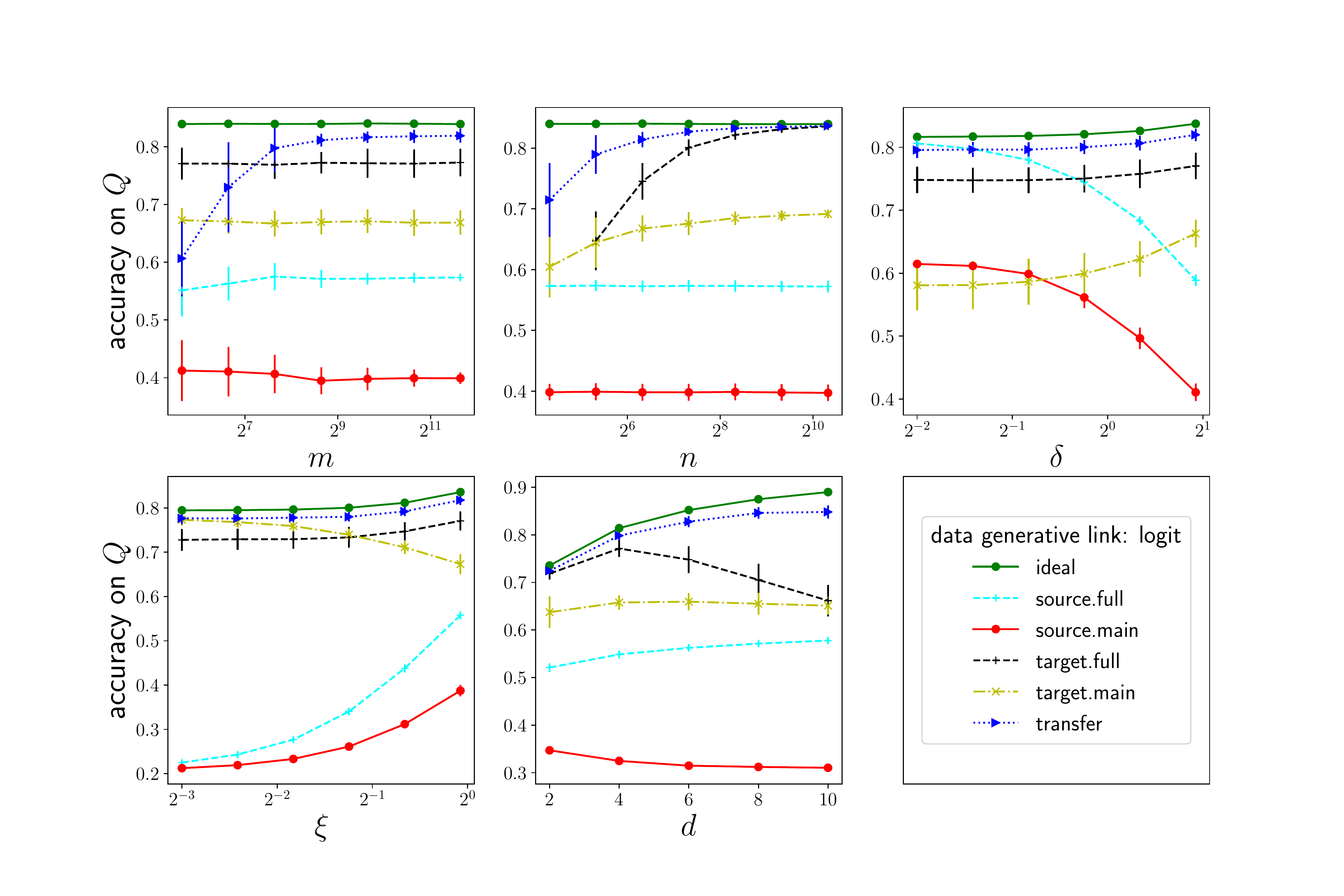}
    \caption{Predictive accuracy in the target domain for different source sample sizes ($m$, \emph{upper left}), target sample sizes ($n$, \emph{upper middle}), similarity between $P$ and $Q$ ($\delta$, \emph{upper right}), non-linearity ($\xi$, \emph{lower left}), and covariate dimension ($d$, \emph{lower middle}). Here the \emph{logit} link is used to generate the labels in both source and target domains. The default values for the parameters are set at $m = 2000, n = 100, \delta = 2, \xi = 1$ and $d = 5$. }
    \label{fig:sim-logit}
\end{figure}

We also considered the effects of link function mis-specification. Here we generate data using the \emph{probit} (inverse CDF of standard normal), but always fit models with the logistic link function. We note that link function mis-specification does not qualitatively change the trends that we observe in Figure \ref{fig:sim-logit}. We conclude that a mis-specified link may give an inconsistent estimate for the coefficients for the data generative model, but it does not detract from the predictive accuracy of the model. In Appendix \ref{sec:appx-simulation}, we present analogous results for the  \emph{cauchit} (inverse CDF of standard Cauchy) and \emph{cloglog} link functions ($\log\{-\log[1-x]\}$).

\begin{figure}
    \centering
    \includegraphics[scale = 0.4]{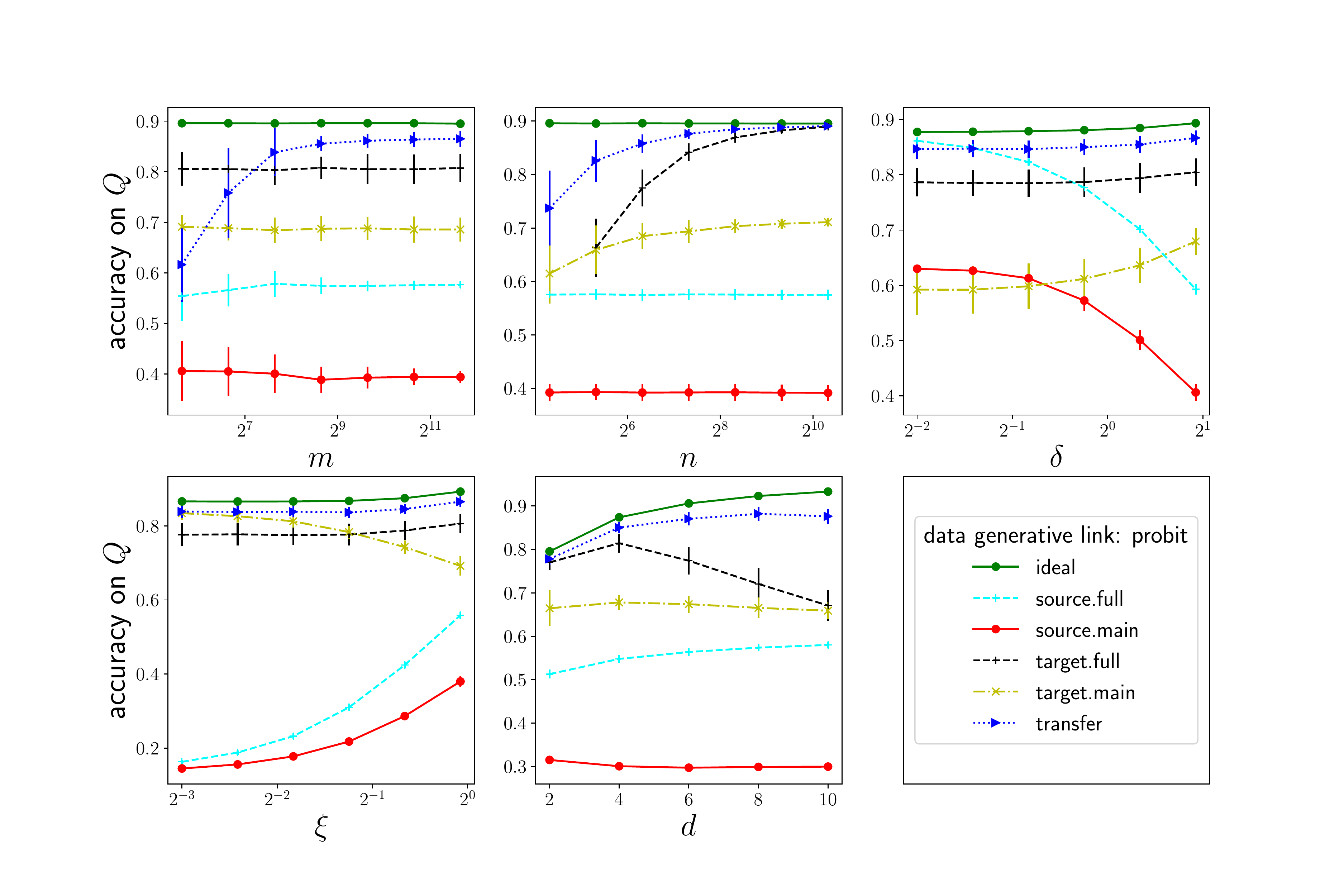}
    \caption{Predictive accuracy in the target $Q$ for different source sample sizes ($m$, \emph{upper left}), target sample sizes ($n$, \emph{upper middle}), similarity between $P$ and $Q$ ($\delta$, \emph{upper right}), non-linearity ($\xi$, \emph{lower left}), and covariate dimension ($d$, \emph{lower middle}). Here the data was generated using \emph{probit} link, but we fit models with the logistic link function. The default values for the parameters are $m = 2000, n = 100, \delta = 2, \xi = 1$ and $d = 5$. }
    \label{fig:sim-probit}
\end{figure}
\section{Predicting 10 year mortality in UK Biobank data}

\label{sec:ukbb} In this section, we study our method's ability to
accurately predict a complex phenotype on data from  the UK Biobank \citep[UKBB;][]{biobank2014uk} repository.
The UKBB contains information from over 500,000 patients aged 40-70
in the United Kingdom, including detailed physical, behavioral, biological
and questionnaire-based outcomes on the participants. Initial enrollment occurred between 2006 and 2010, with ongoing followups.

The use of large-scale datasets to train predictive models of disease
and patient outcomes is currently of great interest in the biomedical
community \citep{Collins2015-np}. Already, such models can accurately
predict risk for diseases such as breast cancer, type 2 diabetes,
and coronary artery disease, in some cases exceeding the performance
of existing diagnostic methods \citep{khera2018genome,tam2019benefits}.
However, the accuracy of these predictions varies according to patient
ancestry, with a bias towards increased accuracy for individuals of
European descent. For example, a recent study found that the accuracy
of a predictive model for schizophrenia declined by more than 50\%
in non-European populations \citep{Vilhjalmsson2015-ap}, and another
observed that certain methods tend to indicate systematically higher
risk across many different diseases in African compared to non-African
populations \citep{Kim2018-oq}. These disparities occur because data
collection efforts have historically focused on genetically European
populations \citep{martin2017human}. Transfer learning has recently
been suggested as a means of reducing ethnicity-based outcome disparities
in clinical applications of machine learning \citep{gao2020deep,li2021targeting,li2021transfer}. 

Here, we aim to predict 10-year all-cause mortality for UKBB participants, \ie,
whether the patient died within 10 years of their first visit to a clinic. A
previous study of this dataset \citep{ganna20155} found that the strongest
overall predictor of mortality was self-reported health
(including age), and that amongst healthy individuals, measures of smoking
habits were the strongest predictors. \cite{ganna20155} also found 
gender-based differences in predictive accuracy: better overall predictive
performance was observed in men, mainly because men die from causes that
are more predictable (e.g., lung cancer), and because of a stronger age effect in
men. The best overall predictor in men was self-reported health, while in women,
it was the existence of a previous cancer diagnosis.

We hypothesize that measurable differences
in predictive accuracy are likely to emerge amongst other subgroups as well. 
The demographic composition of the UKBB broadly mirrors that of the United
Kingdom, and hence consists mainly of individuals of European heritage. 
Indeed, after filtering the dataset to consist of 258,465 unrelated individuals, 
we found that the vast majority ($N=233,810$; 90.46\%) self-identify as ``White''.
We consider this to be the ``source'' population in our model. 
For the ``target'' population we considered all individuals who self-identified
as ``Asian'' ($N=10,699$; 4.14\%). 
We note that the outcome variable
is highly imbalanced: of the $244,509$ (``White'' and ``Asian'') participants included in
our analysis, \textbf{7186 }(\textbf{2.939\%}) died during the study
window. Overall mortality was slightly but significantly lower in
the ``Asian'' population (\textbf{2.337\% }Asian vs. \textbf{2.966\%}
White; $p<10^{-\mathbf{4}}$; $t$-test). In our study we include 34 demographic and health related variables as  predictors of mortality (see Table \ref{tab:ukbb-covariates} in Appendix \ref{sec:appx-ukbb} for further descriptions). 

\subsection{Predicting 10 year mortality with logistic regression}

In Table \ref{tab:ukbb} we present the predictive performance of our linear-adjustment model and two baseline models. We fit a weighted logistic regression to optimize for balanced accuracy because the classes are imbalanced (less than 3\% of the individuals are from the positive class). The weights are set at the inverse of class proportions (of the appropriate subpopulation).  We evaluate the models in terms of their true positive rate (TPR), true negative rate (TNR), AUC and the balanced accuracy ((TPR+TNR)/2). We refer to the Table \ref{tab:ukbb-full} in Appendix \ref{sec:appx-ukbb} for the performance summaries of several other models. 

The two baseline models  that we consider in the main text are logistic regression model with main plus two-way interaction effects fitted on source and target data. Following the terminology in Section \ref{sec:simulation}, we call them  \emph{source.full} and \emph{target.full} respectively.  The linear adjustment model, which we call \emph{transfer}, is described below. 
\begin{equation}
 \begin{aligned}
\hat \eta (x) \triangleq \sigmoid \left(\sum_{j}\hat \alpha_j^{(Q)} x_j + \sum_{j < j'}\hat \beta_{j, j'}^{(P)} x_jx_{j'}\right),\\
\big(\{\hat \alpha_j^{(P)}\}_j, \{\hat \beta_{j,j'}^{(P)}\}_{j<j'}\big) = \argmin \frac1m \sum_{i = 1}^m w_i^{(P)}\ell \left(y^{(P)}_i, \sum_{j} \alpha_j x_j^{(P)} + \sum_{j < j'} \beta_{j, j'} x_j^{(P)}x_{j'}^{(P)}\right)\\
+ \lambda \left(\sum_{j} |\alpha_j| + \sum_{j < j'} |\beta_{j, j'}|\right),\\
\big(\{\hat \alpha_j^{(Q)}\}_j\big) = \argmin \frac1n \sum_{i = 1}^n w_i^{(Q)}\ell \left(y^{(Q)}_i, \sum_{j}  \alpha_j x_j^{(P)} + \sum_{j < j'} \hat \beta_{j, j'}^{(P)} x_j^{(P)}x_{j'}^{(P)} \right)\\
+ \lambda \left(\sum_{j} |\alpha_j| \right)
\,.
\end{aligned}    
\label{eq:offset-interaction}
\end{equation}
The models have 34  main effects and ${34 \choose 2} = 561$  two-way interaction effects. Though the design matrix doesn't quite fall under the regime of high dimensional data ($p \ll n$),  the regularization in our model fitting is required to circumvent  collinearity among several predictors (\eg, weight, body fat percentage and body mass index). 

This approach is motivated by the intuition that the differences between the main effects (in the source and target populations) are larger than the differences between the two-way interaction effects \citep{interactionSame2017,mortalityRace2016,RaceDiffPRS2017Martin}. Thus we elect to refit the main effects to data from the target domain (while keeping the estimates of the interaction effects from the source domain).

Our results are shown in Table \ref{tab:ukbb}. Comparison 
of the \emph{source.full} and \emph{target.full} rows confirms our original hypothesis that mortality prediction is overall more accurate in the White subpopulation. 
However, we see that the AUC and balanced accuracy of the linear adjustment model on Asians is substantially higher than that of the other baseline models; 
in particular, it is substantially higher than naive (\emph{target.full}) model because it has borrowed information from the source population to compensate for the (comparatively) small sample of Asians in the dataset. In fact, in terms of balanced accuracy,  the adjusted model is almost as accurate 
as the ``gold standard'' \emph{source.full} model on Whites. This demonstrates the utility of our approach in improving model accuracy on underrepresented groups.

\begin{table}[]
    \centering
   \begin{tabular}{lllrrrrr}
\toprule
& {models} &   source & target &     TPR &    TNR &    AUC &  bal-acc \\
\midrule
1 & source.full &        White &     White &      0.745 &  0.746 &  0.814 &    0.746 \\
2 & source.full &        White &     Asian &       0.680 &  0.775 &  0.769 &    0.728 \\
3 & target.full &        Asian &     Asian &       0.640 &  0.762 &  0.751 &    0.701 \\
4 & linear.adjustment &  White/Asian &     Asian &       0.740 &  0.750 &  0.773 &    0.745 \\
\bottomrule
\end{tabular}
\caption{Accuracies of 10 year mortality prediction in White and Asian subpopulations under different models. Here ``bal-acc'' is the balanced accuracy,  defined as (TPR+TNR)/2. }
    \label{tab:ukbb}
\end{table}

\subsection{Predicting 10 year mortality with random forest}

In this section we present the details and findings about the 10 year mortality prediction study using random forest, and the utility of our linearly adjusted posterior drift model. The details about the subpopulations and prediction accuracies are same as in the previous subsection. As before we have two baselines \emph{rf.White} and \emph{rf.Asian}, which are random forests fitted on the data from White and Asian subpopulations respectively. Each of these random forests are constructed with 5 million classification trees. To handle the issue of class imbalanced data we take the balanced sampling approach: we sample equal number of individuals from each class to construct each tree, where the samples are drawn with replacement and the sample sizes are equal to the minimum number of observations across classes. This is a standard way to handle imbalanced datasets for random forests \cite{chen2004Using}.
To test the utility of our linearly adjusted approach we consider the \emph{rf.transfer} model. Denoting $\hat p(x)$ as the predicted probability of being in class 1  for an observation with predictor value $x$ in the source   the description of rf.tarnsfer follows.

\begin{equation}
    \begin{aligned}
        \hat \eta (x) = \sigmoid \Big(\logit \big(\hat p(x)\big) + \hat \beta ^\top x\Big),\\
\hat \beta = \argmin \frac 1n \sum_{i = 1}^n \ell \left(y^{(Q)}_i, \logit \big(\hat p(x_i^{(Q)})\big)+  \beta^\top x^{(Q)} \right) + \lambda  \|\beta\|_1\,.
    \end{aligned}
\end{equation}

\begin{table}[]
    \centering
\begin{tabular}{rlllrrrrr}
  \hline
 & models & train-data & test-data &  TPR & TNR & AUC & bal-acc \\ 
  \hline
1 & rf.White & White & White &  0.627 & 0.818 & 0.801 & 0.722 \\ 
  2 & rf.White & White & Asian &  0.500 & 0.842 & 0.814 & 0.671 \\ 
  3 & rf.Asian & Asian & Asian &  0.480 & 0.848 & 0.805 & 0.664 \\ 
  4 & rf.transfer & White/Asian & Asian  & 0.820 & 0.701 & 0.814 & 0.761 \\ 
   \hline
\end{tabular}
    \caption{10 year mortality prediction summaries for White and Asian subpopulations using random forest and linearly adjusted transfer.}
    \label{tab:rf}
\end{table}

We present the summary of model performances in Table \ref{tab:rf}. The conclusions are similar as in the previous subsection. In terms of balanced accuracy, the linear adjustment  approach in row 4 of Table \ref{tab:rf} was able to improve over the two baselines (rf.White and rf.Asian).

\section{Future work}
\label{sec:dsicussion}

In this section, we describe a few possible extensions of our methodology which provide interesting directions for future research in this area.   

\paragraph{Extension to generalized linear models:}

The posterior drift model extends naturally beyond binary classification to the GLM's. Let $\gamma$ denote the link function (e.g. for the Poisson GLM, $\gamma$ is the log function) and $\mu_P(x)$ and $\mu_Q(x)$ be the conditional mean of $Y$ given $X=x$ in the source and target domains respectively:  
\[
\bbE_{\bbP} [Y|X=x] = \mu_P(x), \quad \bbE_{\bbQ} [Y|X=x] = \mu_Q(x).
\] 
As before, for a $d_T$ dimensional transformation $T(x)$ of the feature space, we define the posterior drift as
\begin{equation}
    \gamma(\mu_Q(x)) = \gamma(\mu_P(x)) + \beta^\top T(x).
    \label{eq:posterior-drift-glm}
\end{equation}
Following our method, we estimate the source conditional mean  $\mu_P(x)$ from the (labeled) samples from the source domain (denote this estimator as $\hat \mu_P(x)$) and let $\ell(y, a)$ be the negative log-likelihood for the GLM under consideration with response $y$ and  the link transformed conditional mean $a$. The transfer parameter $\beta$ is estimated via from the following optimization:  
\begin{equation}
\label{eq:optimization-transfer-glm}
    \hat \beta \in \argmin_{\beta} \frac1{n} \sum_{i=1}^n \ell\Big(y_{i}^{(Q)}, \gamma\big(\widehat\mu_P(x_i^{(Q)})\big) +  \beta^\top T\big(x^{(Q)}_{i}\big) \Big)\,.
\end{equation}
Finally, we estimate $\hat \mu_Q(x) := \gamma^{-1}(\gamma(\hat \mu_P(x)) + \hat \beta^T T(x))$, leading to the prediction rule: 
\[ \hat \mu_Q(X^{(Q)}) := \gamma^{-1}(\gamma(\hat \mu_P(X^{(Q)})) + \hat \beta^T T(X^{(Q)}))) \,\]
for a generic target feature realization $X^{(Q)}$.

\paragraph{High dimensional linear shift:}

For a high-dimensional $T(x)$, the linearly shifted posterior drift model can readily accommodate numerous techniques from the vast literature on high dimensional data. For example, one can estimate the coefficient for high dimensional linear shift by adding a regularization term $\cR(\beta)$ in the optimization \eqref{eq:optimization-transfer-glm}:  
\[
\hat \beta \in \argmin_{\beta} \frac1{n} \sum_{i=1}^n \ell\Big(y_{i}^{(Q)}, \gamma\big(\widehat\mu_P(x_i^{(Q)})\big) +  ^\top T\big(x^{(Q)}_{i}\big) \Big) + \cR(\beta)\,. 
\] 
An $\ell_1$-penalization (\cite{tibshirani1996regression}) on the above optimization can be used for model selection on the linear shift. Furthermore, under a sparsity assumption on the true $\beta$, one can establish model selection consistency (similar to lasso consistency \cite[Chapter 6]{buhlmann2011statistics}) for $\hat \beta$.  

\paragraph{Non-parametric estimation of the link function}

In this paper we mainly concern ourselves with prediction on the target data which is known to be robust to mis-specification of the link function. In fact, in our simulations we validated this fact: we showed that if the data generative link differs from the link used for model fitting, the behavior of the prediction error does not differ substantially compared to the true link. But for a mis-specified link, the optimizations \eqref{eq:offset-estimate} and \eqref{eq:optimization-transfer-glm} will typically produce a biased estimate for the linear shift coefficient $\beta$. This could be problematic when one is interested in characterizing the difference between source and target populations and drawings inference on $\beta$. In such cases, the natural approach is to consider non-parametric estimation of the link function where techniques from estimation in single-index models \citep{hristache2001direct,naik2001single,wang2010estimation,kong2007variable, alquier2013sparse} can be adopted.

 \section*{Acknowledgement}
 This paper is based upon work supported by the National Science Foundation (NSF) under grants no.\ 1916271, 2027737,  2113373 (Moulinath Banerjee, Subha Maity and Yuekai Sun) and 2052653 (Jonathan Terhorst). 
 Diptavo Dutta was supported by NIH R01-HG010480-01. The authors would also like to thank Dr. Nilanjan Chatterjee, Bloomberg Distinguished Professor of Biostatistics, Johns Hopkins University, for his insightful comments. 
 This research has been conducted using data from UK Biobank, a major biomedical database, under project ID 42622.
 Any opinions, findings, and conclusions or recommendations expressed in this note are those of the authors and do not necessarily reflect the views of the NSF and NIH.

\bibliographystyle{plainnat}
\bibliography{YK,sm}

\newpage
\appendix
\section{Technical definitions}
\begin{definition}[Support]
\label{def:support}
For a probability measure $P$ defined on $\reals^d$ the support of $P$ is defined as the collection of such points $x\in \reals^d$ that there exists a  $r_x>0$ we have $P\big(B(x, r)\big)>0$ for any $0<r\le r_x$ . 
\end{definition}

\begin{definition}[Taylor polynomial]
\label{def:taylor-poly}
Let  $\bbN_0$ be the set of all non-negative integers. For  an $x=(x_1, \dots, x_d)\in \reals^d$ and an $s = (s_1, \dots, s_d)\in \bbN_0^d$ we define the followings:  $|s| = \sum_{i=1}^d s_i$, $s! = s_1! \dots s_d!$, $x^s =  x_1^{s_1}\dots x_d^{s_d}$ and $D^s = \frac{\partial^{|s|}}{(\partial x_1)^{s_1}\dots (\partial x_d)^{s_d}}$.  For a function $f:\reals^d\to \reals$ and a non-negative integer $b\in \bbN_0$ we define the $b$-th order Taylor polynomial of $f$ at some $x\in \reals^d$ as 
\[
f_x^{(b)}(y) = \sum_{s\in\bbN_0^d, |s|\le b} \frac{(y-x)^s}{s!} D^s f(x), \ y \in \reals^d
\] The polynomial does not exist if any of the differentiation value in the sum does not exist. 
\end{definition}

\section{Proofs}

\subsection{Upper bound}

\begin{proof}[Proof of Theorem \ref{thm:upper-bound}]
\label{proof:ub}

For the simplicity of notations we proceed with $T(x) = x$, but the whole proof is valid for a different $T(x)$. 
We define the followings.  \[
\begin{aligned}
z = (x, y), ~~
\omega(h, \beta; z) = \nabla_{\beta} \ell \big(y; \gamma(h(x)) + \beta^\top x\big) = x \big[- y + \sigma\big(\gamma(h(x)) + \beta^\top x\big) \big],\\
\hat h(x) = \hat \eta_P(x), ~~
h_0(x) = \eta_P(x),~~
\hat \omega(\beta; z) = \omega(\hat h, \beta; z),~~
\omega_0(\beta; z) = \omega(h_0, \beta; z),\\
\hat \bbQ_n = \frac1n \sum_{i=1}^n\delta_{(x_i^{(Q)}, y_i^{(Q)})},~~
M(h, \beta) = \bbQ\big[ \omega(h, \beta; z) \big],~~
M_n(h, \beta) = \hat \bbQ_n\big[ \omega(h, \beta; z) \big]\,.
\end{aligned}
\]

The first order optimality conditions on the estimation of $\beta$ at population and sample levels are 
\[
M_n(\hat h, \hat \beta) = 0 ~~~~ \text{and} ~~~~ M(h_0, \beta_0) = 0.
\]
Note that 
\[
M_n(\hat h, \hat \beta) - M(\hat h, \hat \beta) + M(\hat h, \hat \beta) - M(h_0, \hat \beta) + M(h_0, \hat \beta) - M(h_0, \beta_0)  = 0.
\] We define the followings 
\[
(A) \triangleq M_n(\hat h, \hat \beta) - M(\hat h, \hat \beta), (B) \triangleq M(\hat h, \hat \beta) - M(h_0, \hat \beta), \text{ and }
(C) \triangleq M(h_0, \hat \beta) - M(h_0, \beta_0), 
\]
and  start by proving almost sure convergence for $\hat \beta$. We have established concentration tail bounds (around zero) for $(A)$ (Lemma \ref{lemma:concentration-A}) and hence almost sure convergence to zero (Lemma \ref{lemma:conc-to-as}). Similarly, we have established concentration (around zero) (Lemmas \ref{lemma:concentration-B} and \ref{lemma:concentration-B2}) and almost convergence (to zero) for $(B)$. To deal with $(C)$
we observe that  $E_\bbQ \Big[\ell \big(y; \gamma(h_0(x)) + \beta^\top x\big)\Big] $ is strongly convex  with respect to $\beta$ (Assumption \ref{assump:strong-convexity}) which implies \[
(\hat \beta - \beta_0)^\top (C) \ge f \|\hat \beta - \beta_0\|_2^2\,.
\]
In the inequality
\begin{equation}
\label{eq:mle-first-order}
   0 = (\hat\beta-\beta_0)^\top M_n(\hat h, \hat \beta) \ge f \|\beta - \beta_0\|_2^2 - \|\hat \beta - \beta_0\|_2 \big(\|(A)\|_2 + \|(B)\|_2\big)\,,
 \end{equation}
 we see that $\|A\|_2 + \|B\|_2 \to 0$ almost surely, and hence $\hat \beta \to \beta_0$ almost surely.

Equipped with $\hat \beta \stackrel{\text{a.s.}}{\to}\beta $ we now establish the concentration bound for $\hat \beta$ around $\beta$. 
Lemmas \ref{lemma:concentration-A} and \ref{lemma:concentration-norm} imply that for some $c_A>0$ the inequality  \[
\|(A)\|_2 \le c_A t \sqrt{\frac{d}{n}}
\]  holds with probability at least $1-\frac12e^{-t^2}$. Also, lemmas \ref{lemma:concentration-B} and \ref{lemma:concentration-B2} imply that for $c_B, c_B'>0$  the inequality 
\[
\|(B)\|_2 \le c_B r_m t\sqrt{ \log(n)} e^{c_B' \|\hat \beta\|_2^2 }
\] holds with probability at least $1-\frac12e^{-t^2}$. Using the union probability bound we conclude that the two upper bounds for $\|(A)\|_2$ and $\|(B)\|_2$ hold simultaneously with probability at least $1-e^{-t^2}$. 
Using these upper bounds in the inequality \eqref{eq:mle-first-order} we conclude that  with probability at least $1-e^{-t^2}$
 we have  \[
 \begin{aligned}
 f\|\hat \beta - \beta_0\|_2^2 & \le \|\hat \beta - \beta_0\|_2\left(c_A t \sqrt{\frac{d}{n}} + c_B t r_m \sqrt{\log(n)} e^{c_B' \|\hat \beta\|_2^2 }\right)\\
 & \le \|\hat \beta - \beta_0\|_2\left(c_A t \sqrt{\frac{d}{n}} + c_B t r_m \sqrt{ \log(n)} e^{c_B' \|\hat \beta - \beta_0\|_2^2 + c_B' \|\beta_0\|_2^2 }\right)
 \end{aligned}
  \] or 
  
  \begin{equation}
  \label{eq:convergence-noneq}
      \|\hat \beta - \beta_0\|_2  \le c_1 t \sqrt{\frac{d}{n}} + c_2 r_m t \sqrt{ \log(n)} e^{c_B' \|\hat \beta - \beta_0\|_2^2  }\,.
  \end{equation}
   Fix $t>0$.  For some  $\delta>0$ (independent of $m$ and $ n$) if the inequality  \ref{eq:convergence-noneq} implies  $\|\hat \beta - \beta_0\|_2\ge \delta$, then there exists arbitrarily large $m$ and $n$ such that
   \[
   \Pr \big(\|\hat \beta - \beta_0\|_2\ge \delta\big) \ge 1-e^{-t^2}
   \] 
    This contradicts in probability convergence and hence almost sure convergence for $\hat \beta \to \beta_0$. 
   Therefore, the inequality that $\|\hat \beta - \beta_0\|_2 \le 1$ is true and using this in the right hand side of the  inequality \ref{eq:convergence-noneq}  we have
   \begin{equation}
   \label{eq:concentration-coeff}
      \|\hat \beta - \beta_0\|_2  \le c_3 t \left(\sqrt{\frac{d}{n}} + r_m  \sqrt{ \log(n)} \right)\,, 
   \end{equation}  with probability at least $1-e^{-t^2}$. 
  
Having the concentration for $\hat \beta$ we use lemma \ref{lemma:concentration-target}  to show that  the estimated target regression function concentrates around the true regression function at a rate $\sqrt{{d}/{n}} + r_m  \sqrt{ \log(n)}$. Finally, in lemma \ref{lemma:ub-excess-risk} we let $a_{m, n} = \sqrt{{d}/{n}} + r_m  \sqrt{ \log(n)}$ and conclude that there exists a constant $C>0$ such that 
  \[
  \sup_{(\bbP, \bbQ)} E\big[\cE_Q(\hat f)\big] \le C \left(\sqrt{\frac{d}{n}} + r_m  \sqrt{ \log(n)} \right)\,. 
  \]
\end{proof}

\begin{lemma}
\label{lemma:concentration-A}
There exists a constant  $c>0$ such that for any $h, \beta,$ and $\|a\|_2 = 1$ we have 
\[
\bbQ \Big((\hat \bbQ_n - \bbQ) [a^\top \omega(h, \beta; z)] > t \Big) \le \exp\Big(- \frac{nt^2}{c^2}\Big)\,.
\]
\end{lemma}

\begin{proof}[Proof of lemma \ref{lemma:concentration-A}]
Using the inequalities $\|\omega(h, \beta; z)\|_2 \le \|x\|_2$, $|a^\top \omega(h, \beta; z)| \le |a^\top x|$ and sub-Gaussian tail for $x$ (Assumption \ref{ass:sub-gaussian}) we see the following: for any $u\in \reals^d$
\[
\begin{aligned}
E_\bbQ \Big[\exp\big(u^\top \omega(h, \beta; z) - u^\top E_\bbQ[\omega(h, \beta; z)]\big)\Big] & \le e^{\|u\|_2 E_\bbQ\big[\|x\|_2\big]} E_\bbQ \Big[\exp\big(\big|u^\top \omega(h, \beta; z)\big|\big)\Big]\\
&  \le e^{\|u\|_2 E_\bbQ\big[\|x\|_2\big]} E_\bbQ \Big[\exp\big(\big|u^\top x\big|\big)\Big]\\
& \le e^{\|u\|_2 E_\bbQ\big[\|x\|_2\big]} C_1 e^{C_2\|u\|_2^2 }\\
& \le c_1' e^{c_2'\|u\|_2^2 }
\end{aligned}
\] for some $c_1', c_2'>0$ and $\|u\|_2 \ge u_0.$ Hence, $\omega(h, \beta; z) -  E_\bbQ[\omega(h, \beta; z)]$ also has sub-Gaussian tail. We use Hoeffding's concentration bound to conclude the lemma.  

\end{proof}

\begin{lemma}
\label{lemma:conc-to-as}
For some constant $C>0$ let  the sequence of random variables $\{X_n\}_{n\ge 1}$ to satisfy the concentration inequality that for any $t>0$
\[
P(|X_n| > t) \le e^{-c\big(\frac{t}{r_n}\big)^2}\,.
\] If $r_n\sqrt{\log (n)}\to 0$ then $X_n \to 0$ almost surely.
\end{lemma}

\begin{proof}[Proof of lemma \ref{lemma:conc-to-as}]
Letting $t = \frac{r_n\sqrt{2\log (n)}}{\sqrt{c}}$ in the concentration inequality we see that \[
\sum_{n=1}^\infty P\Bigg(|X_n| > \frac{r_n\sqrt{2\log (n)}}{\sqrt{c}}\Bigg) \le \sum_{n=1}^\infty \frac{1}{n^2} <\infty.
\] Using Borel-Cantelli lemma we have
\[
P\Bigg(|X_n| > \frac{r_n\sqrt{2\log (n)}}{\sqrt{c}}~~ \text{infinitely often}\Bigg) = 0\,.
\] This implies almost sure convergence for $X_n \to 0$. 
\end{proof}

\begin{lemma}
\label{lemma:as-convergence-of-B}
$M(\hat h, \hat \beta) - M(h_0, \hat \beta)$ converges almost surely to zero. 
\end{lemma}

\begin{proof}[Proof of lemma \ref{lemma:as-convergence-of-B}]
We start by  recalling that there exists a constant  $c>0$ such that for $\bbP_X$ almost surely any  $x$ (we denote the $\bbP_X$ probability one set of such $x$'s as $E_P$) and any $t>0$ the following probability bound holds: \[
\bbP \Big(|\hat h(x) - h_0(x)|>t\Big) \le e^{-c\big(\frac{t}{r_m}\big)^2}\,.
\]  Using lemma \ref{lemma:conc-to-as} we conclude that for each $x\in E_P$ there exists a $\bbP$  probability one set $\Omega_x$ such that for any $\omega \in \Omega_x$ we have $\hat h(x; \omega) \to h_0(x)$. For a countable dense subset of $C \subseteq E_P$ the set $\Omega_0 = \cap_{x\in C} \Omega_x$ is again a $\bbP$-probability one set and for any $x \in C$ and $ \omega \in \Omega_0$ we have $\hat h(x; \omega) \to h_0(x)$. 
Since the functions $x\mapsto\hat h(x, \omega), h_0(x)$ are continuous,  we conclude that for any $x \in E_P$ and $ \omega \in \Omega_0$ we have \[
\hat h(x; \omega) \to h_0(x).
\]
We fix a $\hat \beta$. Since the randomness of $\hat \beta$ is independent of the randomness for $\hat h$, fixing a value of $\hat \beta$ has no impact on the randomness of $\hat h$. 
Finally, observing  the facts that $\|\omega(h, \beta; z)\|\le \|x\|$ for any $h$ and $\beta$ and $E_\bbQ\big[\|x\|_2\big]<\infty$ we use dominated convergence theorem to conclude the lemma. 

\end{proof}

\begin{lemma}
\label{lemma:concentration-B}
For some $c, c'>0$ and any $\beta\in \reals^d$ the following inequality holds:
\[
\big\| M(\hat h,  \beta) -  M(h_0,  \beta) \big\|_2 \le c\big\|\hat h(x) - h_0(x)\big\|_{2,   \bbQ_X}e^{c'\|\beta\|_2^2}\,.
\]
\end{lemma}

\begin{proof}[Proof of lemma \ref{lemma:concentration-B}]
We note that 
\[
\begin{aligned}
\omega(\hat h, \beta; z) - \omega(h_0, \beta; z) &= x \big[ \sigma\big(\gamma(\hat h(x)) + \beta^\top x\big) - \sigma\big(\gamma( h_0(x)) + \beta^\top x\big)
 \big]\\
 &= x \left[ \frac{\hat h(x) e^{\beta^\top x}}{\hat h(x) e^{\beta^\top x} + (1-\hat h(x))} - \frac{ h_0(x) e^{\beta^\top x}}{ h_0(x) e^{\beta^\top x} + (1- h_0(x))} \right]\\
 &=  \frac{x(\hat h(x) - h_0(x))e^{\beta^\top x}}{\big\{\hat h(x) e^{\beta^\top x} + (1-\hat h(x))\big\}\big\{h_0(x) e^{\beta^\top x} + (1- h_0(x))\big\}}\,.
\end{aligned}
\]

This follows 

\[
\begin{aligned}
\big\| M(\hat h, \hat \beta) - M(h_0, \hat \beta) \big\|_2 &\le E_{\hat \bbQ} \Big[\big\|\omega(\hat h, \beta; z) - \omega(h_0, \beta; z)\big\|_2\Big]\\
& \le E_{ \bbQ}\Big[ \|x\|_2 \big|\hat h(x) - h_0(x)\big|e^{|\beta^\top x|}\Big]\\
& \le \big\|\hat h(x) - h_0(x)\big\|_{2, \bbQ_X}\Big\| \|x\|_2 e^{|\beta^\top x|}\Big\|_{2, \bbQ_X}\\
& \le \big\|\hat h(x) - h_0(x)\big\|_{2, \bbQ_X}\big\|\|x\|\big\|_{4,  \bbQ_X}\big\|  e^{|\beta^\top x|}\big\|_{4, \bbQ_X}\,,
\end{aligned}
\]
where the last two inequalities are obtained via repeated use of Cauchy-Schwarz inequality. 
Here, from sub-Gaussianity of the covariate $x$ under $\bbQ$ distribution we see that   $\big\|\|x\|\big\|_{4, \bbQ_X}^4 = E_\bbQ \big[\|x\|^4\big] = c_1 <\infty $ and \[
\big\|  e^{|\beta^\top x|}\big\|_{4, \bbQ_X}^4 \le E_\bbQ \big[e^{4|\beta^\top x|}\big] \le c_2 e^{c_3\|\beta\|_2^2}
\] for some $c_1, c_2, c_3>0.$ Using these inequalities we conclude the proof. 

\end{proof}

\begin{lemma}
\label{lemma:inequality-A}
For $u, v\in [0, 1]$ and $b\in \reals$ the following inequality holds, 
\[
\frac{e^b}{\big(ue^b + (1-u)\big)\big(ve^b + (1-v)\big)} \le e^{|b|}\,.
\]
\end{lemma}

\begin{proof}[Proof of lemma \ref{lemma:inequality-A}]

Note that 
\[
\frac{e^b}{\big(ue^b + (1-u)\big)\big(ve^b + (1-v)\big)} = \frac{1}{\big(ue^{\frac b2} + (1-u)e^{-\frac b2}\big)\big(ve^{\frac b2} + (1-v)e^{-\frac b2}\big)}\,.
\] Using the fact that for any $\lambda\in [0, 1]$ and $b\in \reals$
\[
\lambda e^{\frac b2} + (1-\lambda)e^{-\frac b2} \ge e^{-\frac{|b|}{2}}\,,
\] we conclude the proof. 
\end{proof}

\begin{lemma}
\label{lemma:concentration-B2}
Assume that there exists a constant $c>0$ such that for $\bbQ_X$  almost surely any $x$ the followings hold. 
\[
\bbP\big(|\hat h (x) - h_0(x)|>t\big) \le e^{-c\big(\frac{t}{r_m}\big)^2}, t > 0\,.
\]
Then with probability at least $1-\delta$ we have \[
E_\bbQ\big[(\hat h (x) - h_0(x))^2\big] \le c'  {r_m^2}{\log \Big(\frac{2}{\delta}\Big)\log(n)}  \,.
\]
\end{lemma}

\begin{proof}[Proof of lemma \ref{lemma:concentration-B2}]

Since \[E_{\hat \bbQ} \big[(\hat h(x;\omega_P) - h_0(x))^2\big] \to  E_{ \bbQ} \big[(\hat h(x; \omega_P) - h_0(x))^2\big]\] almost surely, we note that there exists a $Q$-probability one set $\Omega^{(Q)}$ such that for $\omega_Q\in \Omega^{(Q)}$ and any point $\omega_P$ from $P$ sample space we have
\[E_{\hat \bbQ} \big[(\hat h(x;\omega_P) - h_0(x))^2\big](\omega_Q) \to  E_{ \bbQ} \big[(\hat h(x; \omega_P) - h_0(x))^2\big] \] as $n\to \infty$. We fix one such $\omega_Q$ and draw our attention to $\omega_P$. This is a pointwise and hence almost sure convergence on $\omega_P$.   Since almost sure convergence implies almost uniform convergence (Egorov's theorem), we have an event $F$ with $P(F) \ge 1-\frac\delta 2$ and the convergence is uniform for $\omega_P\in F$.   Hence, get an $n$ independent of $\omega_P$ such that, for any $\omega_P\in F$
\begin{equation}
\label{eq:tech1}
 E_{\bbQ} \big[(\hat h(x; \omega_P) - h_0(x))^2\big] \le E_{\hat \bbQ} \big[(\hat h(x; \omega_P) - h_0(x))^2\big](\omega_Q) + \frac{\log(2/\delta) r_m^2}{c_1}\,.
\end{equation}
Using the union bound over $\{x_i^{(Q)}(\omega_Q)\equiv x_i^{(Q)} \}$ we get 
\[
\bbP \left( \big|\hat h(x_i^{(Q)}) - h_0(x_i^{(Q)}) \big| > t; i = 1, \dots , n \right)\le n e^{-c\big(\frac{t}{r_m}\big)^2} \le e^{-c\big(\frac{t}{r_m \sqrt{\log (n)}}\big)^2}
\] when $t \ge \frac{r_m\log(n)}{\sqrt{ c(\log(n)-1)}}\ge\frac{r_m(\log(n)-1)}{\sqrt{ c(\log(n)-1)}}\ge c' r_m\sqrt{\log(n)}$. Since, $r_m\sqrt{\log(n)}\to 0$ as $m,n \to \infty$ for suitable $c_1$ 
\[
\bbP \left( \big|\hat h(x_i^{(Q)}) - h_0(x_i^{(Q)}) \big| > t; i = 1, \dots , n \right) \le e^{-c_1\big(\frac{t}{r_m \sqrt{\log (n)}}\big)^2}\,.
\] This is achieved by choosing $c_1$ such that $e^{-c_1c'^2} = 1$. So, with probability at least $1-e^{-c_1\big(\frac{t}{r_m \sqrt{\log (n)}}\big)^2}$ we have \[
E_{\hat \bbQ} \big[(\hat h(x) - h_0(x))^2\big] \le t^2\,. 
\] This implies with probability at least $1-\frac\delta 2$ we have 
\[
E_{\hat \bbQ} \big[(\hat h(x) - h_0(x))^2\big] \le \frac{\log(2/\delta) r_m^2 \log(n)}{c_1}\,.
\] Finally using the inequality \eqref{eq:tech1} (and union bound with the above event and $F$) we conclude the lemma.

\end{proof}

\begin{lemma}
\label{lemma:concentration-norm}

If $x\in \reals^d$ such that there exists a constant $c>0$ such that for any $\|a\|_2 = 1$ we have \[
P\big( a^\top x >t\big) \le e^{-ct^2}, t>0,
\] then we have \[
P\big(\|x\|_2>t\big) \le e^{-\frac{c't^2}{d}}, t>0\,, 
\]
for some $c'>0$
\end{lemma}

\begin{proof}[Proof of lemma \ref{lemma:concentration-norm}]

Note that

Let $F$ is a $1/2$-cover of $\cS^{d-1}$ (the set of all norm one vectors), \ie, \ $\{x\in \reals^d: \|x-y\|_2\le 1/2 \text{ for some } y \in F\}\supseteq \cS^{d-1}$ and for any $x, y\in F$ we have $\|x-y\|_2 \ge 1/4$. Then $|F|\le 5^d$.

Note that for any $a\in \cS^{d-1}$ we can find $b\in F$ such that $\|a-b\|\le 1/2.$ This follow
\[
a^\top x \le b^\top x + \|a-b\|_2 \|x\|_2 \le b^\top x + \frac{\|x\|_2}{2} 
\]
or, \[
\|x\| \le \sup_{b \in F} b^\top x +  \frac{\|x\|_2}{2} \implies \|x\| \le 2\sup_{b \in F} b^\top x
\] Hence, 
\[
P\big(\|x\|>t\big)\le P\Big(\sup_{b\in F} b^\top x >t/2\Big) \le 5^d e^{-c_1t^2}
\] 
For $t\ge \sqrt{\frac{2c_2 d}{c_1}}$ we have
\[
P\big(\|x\|>t\big) \le  e^{c_2 d -c_1t^2} \le e^{-\frac{c_1t^2}{2d}}
\] Letting $\alpha = t \sqrt{d}$ we have \[
P\big(\|x\|>\alpha \sqrt{d}\big) \le   e^{-\frac{c_1\alpha^2}{2}}, t \ge  \sqrt{\frac{2c_2}{c_1}}
\] We choose, $c'$ such that \[
P\big(\|x\|>\alpha \sqrt{d}\big) \le   e^{-{c'\alpha^2}}, t >0\,.
\] This concludes the lemma. 
\end{proof}

\begin{lemma}
\label{lemma:tech2}
For any $h$ and $\beta$ the following holds. 
\[
\Big|\sigma\big(\gamma(h(x)) + \beta^\top x\big) - \sigma\big(\gamma(h_0(x)) + \beta_0^\top x\big)\Big| \le |h(x) - h_0(x)| e^{\big(\|\beta - \beta _0\|_2 + \|\beta_0\|_2\big)\|x\|_2}  + \big|(\beta -\beta_0)^\top x\big|\,.
\]
\end{lemma}

\begin{proof}[Proof of lemma \ref{lemma:tech2}]

We start by defining \[
f(a, b) \triangleq \sigma\big(\gamma(a) + b\big) = \frac{ae^b}{ ae^b + (1-a)}\,.
\]
Note the followings. 
\[
\begin{aligned}
\partial_a f(a,b) &= \frac{e^b}{\big(ae^b + (1-a)\big)^2}\\
& = \frac{1}{\Big(ae^{\frac b2} + (1-a)e^{-\frac b2}\Big)^2} \le e^{|b|}\,,
\end{aligned}
\] and 
\[
\partial_b f(a, b) = \sigma\big(\gamma(a) + b\big) \Big(1-\sigma\big(\gamma(a) + b\big)\Big) \le \frac14\,. 
\]

From the first order Taylor series expansion we get 
\[
f(a, b) - f(a_0, b_0) = (a-a_0) \partial_a f(\tilde a, \tilde b) + (b-b_0) \partial_b f(\tilde a, \tilde b)\,,
\] where $(\tilde a, \tilde b) = \lambda (a, b) + (1-\lambda) (a_0, b_0)$ for some $0\le \lambda\le 1. $
Hence, we get 
\[
|f(a, b) - f(a_0, b_0)| \le |a-a_0| e^{|\tilde b|} + |b-b_0|\,.
\] Here, $a = h(x), a_0 = h_0(x), b =  theta^\top x $ and $b_0 = \beta_0^\top x$. In this case \[
\begin{aligned}
|\tilde b| &= |\lambda \beta^\top x + (1-\lambda) \beta_0^\top x| \\
& \le \lambda |(\beta - \beta_0)^\top x| + |\beta_0^\top x|\\
& \le \big(\|\beta - \beta _0\|_2 + \|\beta_0\|_2\big)\|x\|_2\,.
\end{aligned}
\] This concludes the lemma. 
\end{proof}

\begin{lemma}
\label{lemma:concentration-target}
There exists $c_1, c_2>0$ such that for any $\bbQ_X$-almost sure $x$ with probability at least $1 -e^{-t^2}$ we have 
\[
|\hat \eta_Q(x) - \eta_Q(x)| \le c_1 t \Big( r_m\log(n) + \sqrt{\frac dn } \Big) e^{c_2\|x\|_2}\,,
\] for $t >0$.  
\end{lemma}

\begin{proof}
We note that \[
\begin{aligned}
|\hat \eta_Q(x) - \eta_Q(x)| & = \Big|\sigma\big(\gamma(\hat h(x)) + \beta_0^\top x\big) - \sigma\big(\gamma(h_0(x)) + \beta_0^\top x\big)\Big| \\
& \le |\hat h(x) - h_0(x)| e^{\big(\|\hat \beta - \beta _0\|_2 + \|\beta_0\|_2\big)\|x\|_2}  + \big|(\hat \beta -\beta_0)^\top x\big|
\end{aligned}
\] from lemma \ref{lemma:tech2}. From Assumption \ref{assumption:source-reg-fn-concentration} (concentration of source regression function) we see that there exists a $c_h>0$ such that for any $\bbQ_X$-almost sure $x$ with probability at least $1 - \frac12 e^{-t^2}$ we have $|\hat h(x) - h_0(x)| \le c_h t r_m$. Fix such  an $x$. From equation \eqref{eq:concentration-coeff} we see that  there exists a constant $c_\beta>0$ such that with probability at least  $1 - \frac12 e^{-t^2}$ we have $\|\hat \beta - \beta_0\|\le c_\beta t (r_m\log(n) + \sqrt{d/n}).$  Using union bound over the two probability bounds before, we see that  with probability at least $1-e^{-t^2}$ the inequalities 
\[
|\hat h(x) - h_0(x)| \le c_h t r_m, ~~~ \|\hat \beta - \beta_0\|\le c_\beta t (r_m\log(n) + \sqrt{d/n})\,,
\] hold simultaneously
and under this event we have 
\[
\begin{aligned}
|\hat \eta_Q(x) - \eta_Q(x)| 
& \le |\hat h(x) - h_0(x)| e^{\big(\|\hat \beta - \beta _0\|_2 + \|\beta_0\|_2\big)\|x\|_2}  + \big|(\hat \beta -\beta_0)^\top x\big|\\
& \le c_ht r_m e^{\Big(c_\beta t (r_m\log(n) + \sqrt{d/n}) + \|\beta_0\|_2 \Big)\|x\|_2} + c_\beta t \Big( r_m\log(n) + \sqrt{\frac dn } \Big)\|x\|_2\\
& \le \left[c_ht r_m + \frac{c_\beta t \Big( r_m\log(n) + \sqrt{\frac dn } \Big)}{c_\beta t \Big( r_m\log(n) + \sqrt{\frac dn } \Big) + \|\beta_0\|_2}\right] e^{\Big(c_\beta t (r_m\log(n) + \sqrt{d/n}) + \|\beta_0\|_2 \Big)\|x\|_2}\\
& ~~~ \left(\text{This is obtained using } x \le e^x \text{ for }x>0.\right)\\
& \le c_1 t \Big( r_m\log(n) + \sqrt{\frac dn } \Big) e^{c_2\|x\|_2}\,,
\end{aligned}
\] for $t \le  \frac{c_3}{r_m\log(n) + \sqrt{\frac dn }}$ and some $c_1, c_2, c_3>0$.  Note that this event depends on $x$ whereas the constants does not. 

At $t = \frac{c_3}{r_m\log(n) + \sqrt{\frac dn }}$ we see that the upper bound is $c_1c_3e^{c_2\|x\|_2}\ge c_1c_3$. We replace $c_1$ by $c_1 \vee \frac 1{c_3}$  which again  gives us a valid upper bound and after such replacement we have $c_1c_3\ge 1$. Since $|\hat \eta_Q(x) - \eta_Q(x)| \le 1$ we have the inequality for all $t>0$. This concludes the proof. 

\end{proof}

\begin{lemma}
\label{lemma:ub-excess-risk}
Let us assume there exists $c_1, c_2>0$ and a sequence $\{a_{m,n}\}_{m, n\in \bbN}$  satisfying that $a_{m,n} \to 0$ as $m, n \to \infty$ such  that  for any $\bbQ_X$-almost sure $x$ with probability at least $1 -e^{-t^2}$ we have 
\[
|\hat \eta_Q(x) - \eta_Q(x)| \le c_1 t a_{m,n} e^{c_2\|x\|_2}\,,
\] for $t >0$. Then under the margin condition with parameter $\alpha$ (Assumption \ref{assump:margin-condition}) there exists a constant $C>0$ such that the following holds. 
\[
\sup_{(\bbP, \bbQ)} E \big[\cE_Q(\hat f)\big] \le C a_{m, n}^{1+\alpha}\,. 
\]
\end{lemma}

\begin{proof}[Proof of lemma \ref{lemma:ub-excess-risk}]

This proof of the lemma \ref{lemma:ub-excess-risk} is a modification on the proof of \cite[Lemma 3.1]{audibert2007fast}. For some $\delta > 0$ we define the following sets. 
\[
\begin{aligned}
A_0 & \triangleq \left\{x: 0<\big|\eta_Q(x) - \frac12\big|<\delta \right\},\\
A_j & \triangleq \left\{x: 2^{j-1}\delta<\big|\eta_Q(x) - \frac12\big|<2^j \delta \right\}, j \ge 1\, .
\end{aligned}
\]
Defining $f_Q$ as the Bayes classifier for target we get
\[
\begin{aligned}
E \big[\cE_Q(\hat f)\big] & = E\Big[E_{\bbQ_X}\big[|2\eta_Q(x)-1| \bbI\{\hat f \neq f_Q\} \big]\Big]\\
& = \sum_{j \ge 0} E\Big[E_{\bbQ_X}\big[|2\eta_Q(x)-1| \bbI\{\hat f \neq f_Q\} \bbI \{x \in A_j \} \big]\Big]\\
& \le 2\delta \bbQ_X \Big(0 < \big|\eta_Q(x) - \frac1 2\big|<\delta \Big)\\
& ~~~~ + \sum_{j \ge 1} E\Big[E_{\bbQ_X}\big[|2\eta_Q(x)-1| \bbI\{\hat f \neq f_Q\} \bbI \{x \in A_j \} \big]\Big]\\
& \le 2c_\alpha\delta^{1+\alpha}  + \sum_{j \ge 1} E\Big[E_{\bbQ_X}\big[|2\eta_Q(x)-1| \bbI\{\hat f \neq f_Q\} \bbI \{x \in A_j \} \big]\Big]\,. 
\end{aligned}
\] On the event $\{\hat f \neq f_Q\}$ we see that the inequality $|\eta_Q(x)-1/2| \le |\hat \eta_Q(x) - \eta_Q(x)|$ is satisfied. So for any $j \ge 1$, we get 
\[
\begin{aligned}
& E\Big[E_{\bbQ_X}\big[|2\eta_Q(x)-1| \bbI\{\hat f \neq f_Q\} \bbI \{x \in A_j \} \big]\Big]\\
&   ~~~ \le 2^{j + 1}\delta E\Big[E_{\bbQ_X}\big[ \bbI\{\hat f \neq f_Q\} \bbI \{x \in A_j \} \big]\Big]\\
& ~~~ \le 2^{j + 1}\delta E_{\bbQ_X} \Big[ P\big( |\hat \eta_Q (x) - \eta_Q(x)| > 2^{j - 1}\delta \big) \bbI \big\{ 0 < |\eta_Q(x) - 1/2| < 2^j \delta \big\} \Big] \\
& ~~~ \le2^{j + 1}\delta E_{\bbQ_X} \left[ e^{- \frac{2^{2j - 2}\delta^2}{c_1^2 a_{m, n}^2 e^{2c_2\|x\|_2}}} \bbI \big\{ 0 < |\eta_Q(x) - 1/2| < 2^j \delta \big\} \right]\\
& ~~~ = 2^{j + 1}\delta E_{\bbQ_X}\left[ e^{- \frac{2^{2j - 2}\delta^2}{c_1^2 a_{m, n}^2 e^{2c_2\|x\|_2}}} \bbI \big\{ 0 < |\eta_Q(x) - 1/2| < 2^j \delta \big\} \bbI \big\{\|x\|_2  \le \chi  \big\} \right]\\
& ~~~~~~ + 2^{j + 1}\delta E_{\bbQ_X}\left[ e^{- \frac{2^{2j - 2}\delta^2}{c_1^2 a_{m, n}^2 e^{2c_2\|x\|_2}}} \bbI \big\{ 0 < |\eta_Q(x) - 1/2| < 2^j \delta \big\} \bbI \big\{\|x\|_2 > \chi  \big\} \right]\\
& ~~~\le 2^{j + 1}\delta e^{- \frac{2^{2j - 2}\delta^2}{c_1^2 a_{m, n}^2 e^{2c_2\chi }}} \bbQ_X \big(0 < |\eta_Q(x) - 1/2| < 2^j \delta\big) + 2^{j + 1}\delta \bbQ_X\big(\|x\|_2 > \chi\big)\\
& ~~~ \le 2\cdot  2^{j(1+\alpha)}\delta^{1 + \alpha} e^{- \frac{2^{2j - 2}\delta^2}{c_1^2 a_{m, n}^2 e^{2c_2\chi }}}  + 2^{j + 1}\delta e^{-c_3\chi^2}\,. 
\end{aligned}
\]
Here, we note that $|\eta_Q(x) - 1/2|\le 1/2$ which implies for any $j$ with $2^{j}\delta \ge 1$ we have \[
 E\Big[E_{\bbQ_X}\big[|2\eta_Q(x)-1| \bbI\{\hat f \neq f_Q\} \bbI \{x \in A_j \} \big]\Big] = 0\,.
\] So, for $j \ge 1$ we have  \[
E\Big[E_{\bbQ_X}\big[|2\eta_Q(x)-1| \bbI\{\hat f \neq f_Q\} \bbI \{x \in A_j \} \big]\Big] \le 2\cdot  2^{j(1+\alpha)}\delta^{1 + \alpha} e^{- \frac{2^{2j - 2}\delta^2}{c_1^2 a_{m, n}^2 e^{2c_2\chi }}}  + 2  e^{-c_3\chi^2}\,. 
\]
We let $\delta = a_{m, n}$. Then we have
\[
\begin{aligned}
 E\Big[E_{\bbQ_X}\big[|2\eta_Q(x)-1| \bbI\{\hat f \neq f_Q\} \bbI \{x \in A_j \} \big]\Big] \le 2\cdot  2^{j(1+\alpha)}\delta^{1 + \alpha} e^{- \frac{2^{2j - 2}}{c_1^2  e^{2c_2\chi }}}  + 2 e^{-c_3\chi^2}\,.
\end{aligned}
\] We fix $\chi = \sqrt{\frac{(1+\alpha)\log(1/\delta) + j \frac{c_3(\log(2))^2}{4c_2^2}}{c_3}}$. For such $\chi$ we have $e^{-c_3\chi^2} = \delta^{1+\alpha}e^{-j \frac{c_3(\log(2))^2}{4c_2^2}} = \delta^{1+\alpha} e^{-c_4j }$ and \[
 \chi \le \frac{\sqrt{(1+\alpha)\log(1/\delta)} + j \sqrt{c_3} \frac{\log(2)}{2c_2}}{\sqrt{c_3}}
\] 
or \[
e^{2c_2\chi } \le e^{\sqrt{\frac{(1+\alpha)\log(1/\delta)}{c_3}}} 2^j  = c_5 2^j 
\] 
for some constant $c_4, c_5>0$. 
Hence for an appropriate $c_6, c_7>0$ we have 
\[
\begin{aligned}
 & \sum_{j \ge 1} E\Big[E_{\bbQ_X}\big[|2\eta_Q(x)-1| \bbI\{\hat f \neq f_Q\} \bbI \{x \in A_j \} \big]\Big] \\
& ~~~ \le \sum_{j \ge 1} \Big[2\cdot  2^{j(1+\alpha)}\delta^{1 + \alpha} e^{-c_6 2^j  }  + 2 \delta^{1+\alpha} e^{-c_4j }\Big]\\
& ~~~ \le c_7 \delta^{1+\alpha}\,. 
\end{aligned}
\] 
This concludes the lemma. 
\end{proof}

\subsection{Lower bound}

\label{sec:proof-lb}

\begin{proof}[Proof of Theorem \ref{thm:lower-bound}]
We break the proof in two part. The first part is related to the non-parametric estimation of $\eta_P$. We refer to the proof of \cite[Theorem 3.5]{audibert2007fast} for the proof  of this part. To directly relate to the proof we let $\bbP = \bbQ$. In this case the source and target satisfies transfer with $\beta = 0$. This leads to a bound \[
\inf_{f}\sup_{(\bbP, \bbQ)}E\big[\cE_{\bbQ_X}(f)\big] \ge \inf_{f}\sup_{\substack{(\bbP, \bbQ),\\ \bbP = \bbQ}}E\big[\cE_{\bbQ_X}(f)\big]  \ge c' (m+n)^{-\frac{\beta(1+\alpha)}{2\beta+d_F}}
\]
 for some constant $c'>0$ which is independent of $\alpha, \beta, m$ and $m$. Noticing $m \ge n$ we have the bound \[
\inf_{f}\sup_{(\bbP, \bbQ)}E\big[\cE_{\bbQ_X}(f)\big]  \ge c m^{-\frac{\beta(1+\alpha)}{2\beta+d_F}}
\] for some constant $c>0$.

\paragraph{Parametric part}
For simplicity we  call $d_T $ by $d$ and let $T(x) = x$. The proof easily adapts to other $T(x)$. 
Here we shall create two probability distributions $H_1$ and $H_{-1}$ and apply LeCam's lower bound. For each of these probability distributions we fix the followings. 
\[
\eta_P(x) = \frac12, \ p_X(x) = 1,\ x \in [0, 1]^d. 
\]
Note that the assumptions about regular support and  strong density of $\bbP_X$ (Assumption \ref{assump:strong-density}) and smoothness on $\eta_P$ (Assumption \ref{assump:smoothness}) are satisfied in the above specifications.

Define $d_0 = \lfloor d/4\rfloor$ and $r = \sqrt{d/n}$. 
Let $u \in \bD \triangleq \{-1, 0, 1\}^d$. We define a probability measure on $\bD$ in the following way. With $A(u) = \sum_{j = 1}^{2d_0} u_j$ and $B(u) = \sum_{j = 1 + 2d_0}^{4d_0} u_j$
we define $P(A(u) = a, B(u) = b)$ in  Table \ref{tab:prob}.
Within each these events we assign equal probability mass  to  each of the points. This induced a probability on the set $ \tilde D = \{1/2 + u/4: u \in D\}$ where the scalar addition and multiplication to a vector is interpreted in the usual sense. We call this probability $\bbQ_X$. Note that $\bbQ_X$ is a discrete probability with support points in $[0, 1]^d$. This probability satisfies the Assumption \ref{assump:dominated-measure} that support of $\bbP_X$ contains the support of $\bbQ_X$. \begin{table}
     \centering
    \begin{tabular}{ccccc}
    \toprule
         & A(u) & -2 & 0 & 2  \\
        B(u) & &&&\\
        \midrule
        -2 & & $\frac{c_\alpha r^\alpha}{2}$ & 0 & 0\\
        0 & &  0 & $1 - c_\alpha r^\alpha$ & 0\\
        2 && 0 & 0 & $\frac{c_\alpha r^\alpha}{2}$\\
        \bottomrule
        
    \end{tabular}
 \caption{Probabilites }
   \label{tab:prob}
\end{table}
For $\nu \in \{-1, 1\}$ we define 
\[
\beta_\nu = r \big( 1_{2d_0}^\top , \nu 1_{2d_0}^\top, 0_{d - 4d_0} \big)\in \reals^d, ~~ \eta_Q^{(\nu)}(x) = \sigma \left(\beta_\nu^\top (x-1/2)\right)\,.
\] 
Note that $\beta_\nu^\top (x-1/2) = \frac{r} 4\big(A(u) + \nu B(u)\big)$. 
In lemma \ref{lemma:margin-lb-par2} we show that $\bbQ$ satisfies margin condition. We define the probabilities $H_\nu = \bbP^{\otimes m} \otimes \bbQ_\nu ^{\otimes n}$. Next we call the Bayes classifier for $\bbQ_\nu$ as $f_\nu$ and calculate the excess risk of $f_{-1}$ under the probability $\bbQ_1. $ Note that 
\[
\begin{aligned}
\cE_{\bbQ_1}(f_{-1}) &=  E_{\bbQ_{X, 1}}\Big[\big|\eta_Q^{(1)}(x) - 1/2  \big|\bbI \big\{f_1(x) \neq f_{-1}(x)\big\}\Big]
\end{aligned}
\]
Here, denoting $\beta_\nu^\top (x-1/2)$ by $v$ we see that 
\[
\begin{aligned}
\Big|\eta_Q^{(\nu)}(x) - \frac12 \Big| & = \frac{\big|e^{v} -1\big|}{2(e^{v} + 1)}
 \ge \frac{\big|e^{v} -1\big|}{2} \ge \frac{v}{4}
\end{aligned}
\]
as long as $|v| \le r \le  \log (2)$. This implies
\[
\begin{aligned}
\cE_{\bbQ_1}(f_{-1}) &\ge  \frac{r} 4 E_{\bbQ_{X, 1}}\Big[\big|A(u) + B(u)   \big|\bbI \big\{A(u) + B(u)  \neq A(u) - B(u) \big\}\Big]\\
& = \frac{r} 4 E_{\bbQ_{X, 1}}\Big[\big|A(u) + B(u)   \big|\bbI \big\{ B(u)  \neq 0 \big\}\Big]\\
& = \frac{r} 4 \left[2 \times 4 \frac{c_\alpha r^\alpha}{2}  \right]\\
& \ge c' r^{1+\alpha} , 
\end{aligned}
\] for some $c'>0$.

To calculate the Kulback-Leibler divergence between $H_{1}$ and $H_{-1}$ we notice that 
\[
\begin{aligned}
\text{KL} (H_{1}|| H_{-1}) &= n \text{KL}(\bbQ_{1} || \bbQ_{-1})\\
& = n E_{\bbQ_{X, 1}} \left[\eta_{Q}^{(1)}(x) \log\Bigg(\frac{\eta_{Q}^{(1)}(x)}{\eta_{Q}^{(-1)}(x)}\Bigg) +  \big(1-\eta_{Q}^{(1)}(x)\big) \log\Bigg(\frac{1-\eta_{Q}^{(1)}(x)}{1-\eta_{Q}^{(-1)}(x)}\Bigg) \right]\\
\end{aligned}
\]

Letting $A(x) = \sum_{j = 1}^{2d_0} (x_j - 1/2) = A(u)/4$ and $B(x) = \sum_{j = 1 + 2d_0}^{4d_0} (x_j - 1/2) = B(u)/4$ we see that 
\[
\begin{aligned}
& \eta_{Q}^{(1)}(x) \log\Bigg(\frac{\eta_{Q}^{(1)}(x)}{\eta_{Q}^{(-1)}(x)}\Bigg) +  \big(1-\eta_{Q}^{(1)}(x)\big) \log\Bigg(\frac{1-\eta_{Q}^{(1)}(x)}{1-\eta_{Q}^{(-1)}(x)}\Bigg)\\
& ~~= \sigma\big(A(x) + B(x)\big) \log\Bigg(\frac{\sigma\big(A(x) + B(x)\big)}{\sigma\big(A(x) - B(x)\big)}\Bigg) \\
& ~~~~~~+ (1-\sigma)\big(A(x) + B(x)\big) \log\Bigg(\frac{(1-\sigma)\big(A(x) + B(x)\big)}{(1-\sigma)\big(A(x) - B(x)\big)}\Bigg)\\
& ~~ = (\bA)
\end{aligned}
\] Since $A(x), B(x) \to 0$ on the support, we see that the above quantity is approximately 
\[
\begin{aligned}
(\bA) & \sim \big(A(x) + B(x)\big) \big(A(x) + B(x) - A(x) + B(x)\big) \\
& ~~~ + \big(-A(x) - B(x)\big) \big(-A(x) - B(x) + A(x) - B(x)\big)\\
& = 4B(x) \big(A(x) + B(x)\big) = \frac{r^2}{4} B(u) \big(A(u) + B(u)\big)
\end{aligned}
\]
Using the asymptotic behavior on $(\bA)$ we have \[
\begin{aligned}
\text{KL} (H_{1}|| H_{-1}) & \sim \frac{nr^2}{4} E_u \Big[B(u) \big(A(u) + B(u)\big)\Big]\\
& = \frac{nr^2}{4} \frac{c_\alpha r^\alpha}{2} = c_1 d r^{\alpha} \le 1
\end{aligned}
\] for sufficiently large $n$ 
since $r\to 0$ in fixed dimensional setting.

From Equation (2.5) and Theorem 2.2 (iii) in \cite{tsybakov2009Introduction} we have the formulation that 
\[
\min_f \max_{(\bbP, \bbQ)} E\big[\cE_{\bbQ_X}(f)\big] \ge \frac{\cE_{\bbQ_1}(f_{-1})}{8} e^{-\text{KL} (H_{1}|| H_{-1})}
\] where letting $E\big[\cE_{\bbQ_X}(f)\big] \ge c' r^{1+\alpha}$ and $\text{KL} (H_{1}|| H_{-1}) \le 1$ we 
conclude that there exists a constant $c>0 $ such that \[
\min_f \max_{(\bbP, \bbQ)} E\big[\cE_{\bbQ_X}(f)\big] \ge c\Big(\frac dn\Big)^{\frac{1+\alpha}{2}}.
\]

Finally we combine the lower bounds for non-parametric and parametric parts to obtain 
\[
\min_f \max_{(\bbP, \bbQ)} E\big[\cE_{\bbQ_X}(f)\big] \ge c \Big(\frac {d_T} n\Big)^{\frac{1+\alpha}{2}} \vee m^{-\frac{\beta(1+\alpha)}{2\beta+d_F}}
\] for some constant $c>0$ or \[
\min_f \max_{(\bbP, \bbQ)} E\big[\cE_{\bbQ_X}(f)\big] \ge c' \left(\Big(\frac {d_T} n\Big)^{\frac{1+\alpha}{2}} + m^{-\frac{\beta(1+\alpha)}{2\beta+d_F}}\right)^{1+\alpha}
\] for some constant $c' > 0$. 

\end{proof}

\begin{lemma}
\label{lemma:margin-lb-par2}
The regression functions $\eta_Q^{(\nu)}$ satisfies $\alpha$ margin condition with respect to $\bbQ_X$. 
\end{lemma}

\begin{proof}[Proof of lemma \ref{lemma:margin-lb-par2}]
Note that \[
\bbQ_X\left(0 < \Big|\eta_Q^{(\nu)}(x) - \frac12 \Big|\le  t \right) \le c_\alpha r^\alpha. 
\] 
Note that \[
t \ge \Big|\eta_Q^{(\nu)}(x) - \frac12 \Big| = \frac{\big|e^u -1\big|}{2(e^u + 1)} \ge \frac{\big|e^u -1\big|}{2}
\] which follows \[
-r < \log(1- 2t) \le x \le \log(1+ 2t) < r
\] if $2t \le \big(e^r - 1\big) \wedge \big( 1- e^{-r}\big) \le r$ as long as $|r|\le \log(2)$. 

Hence, 
\[
\begin{aligned}
\bbQ_X\left(0 < \Big|\eta_Q^{(\nu)}(x) - \frac12 \Big|\le  t \right) &\le \bbQ_X\left(0 < \left|\sum_{j=1}^{2d_0} u_j +\nu \sum_{j=2d_0 + 1}^{4d_0} u_j \right|\le  c_2t \right)\\
& = \begin{cases}
= 0 & t < 2r,\\
\le c_\alpha r^\alpha & t \ge 2r.
\end{cases}
\end{aligned}
\] For an appropriate choice of $c_\alpha$ this implies 

\[
\bbQ_X\left(0 < \Big|\eta_Q^{(\nu)}(x) - \frac12 \Big|\le  t \right) \le C_\alpha t^\alpha\,. 
\]

\end{proof}

\section{Experimental supplements}

Codes for the experiments can be found in the \href{https://github.com/smaityumich/linearly-shifted-transfer}{github repository}. 

\subsection{Synthetic experiments}
\label{sec:appx-simulation}

\begin{figure}
    \centering
    \includegraphics[scale = 0.35]{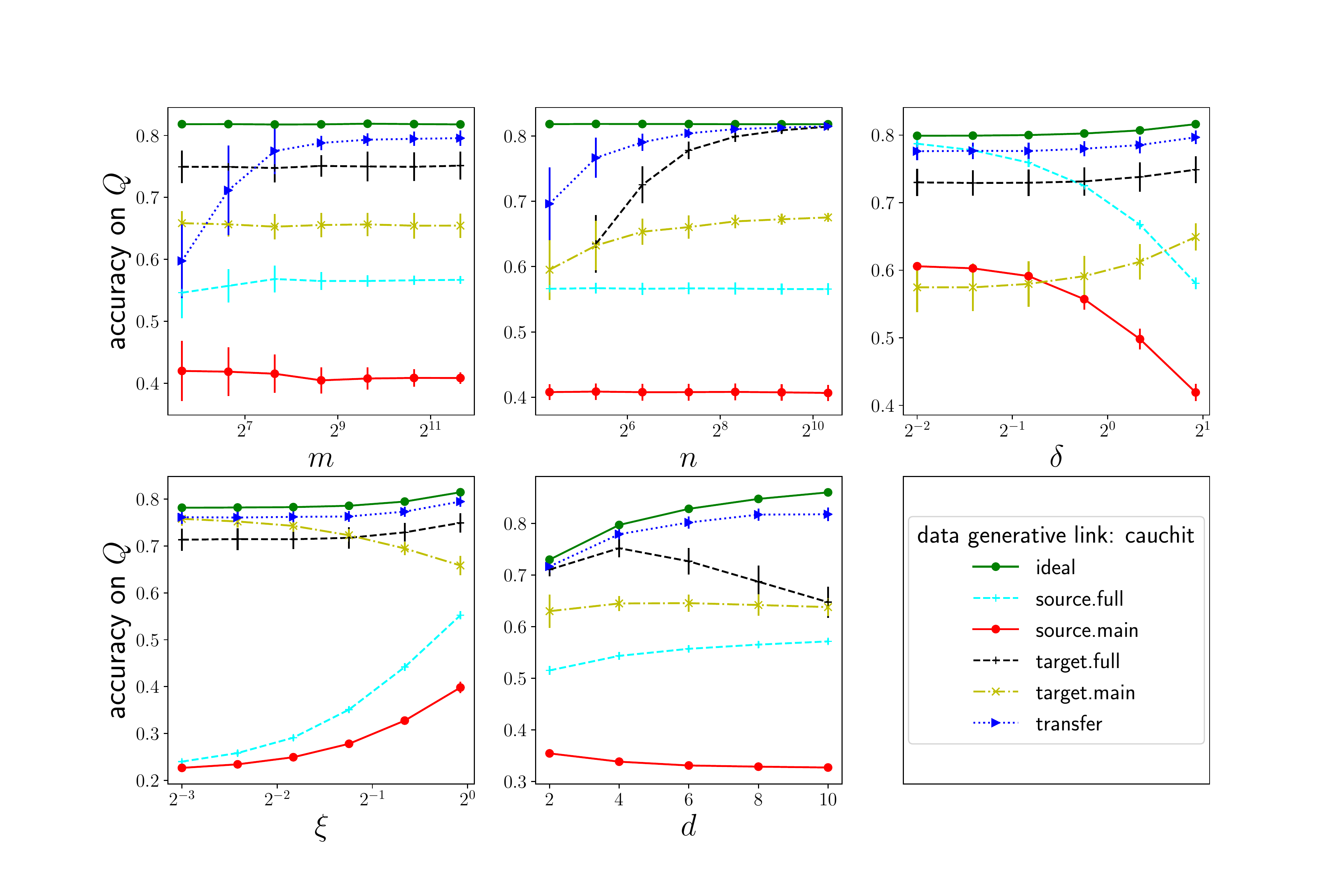}
    \caption{Predictive accuracy on target $Q$ for different source sample size ($m$, \emph{upper left}), target sample size ($n$, \emph{upper middle}), linear separation between $P$ and $Q$ ($\delta$, \emph{upper right}), non-linearity strength ($\xi$, \emph{lower left}) and covariate dimension ($d$, \emph{lower middle}). The \emph{cauchit} link is used in data generative process. The default values for this parameters are set at $m = 2000, n = 100, \delta = 2, \xi = 1$ and $d = 5$. }
    \label{fig:sim-logit}
\end{figure}

\begin{figure}
    \centering
    \includegraphics[scale = 0.35]{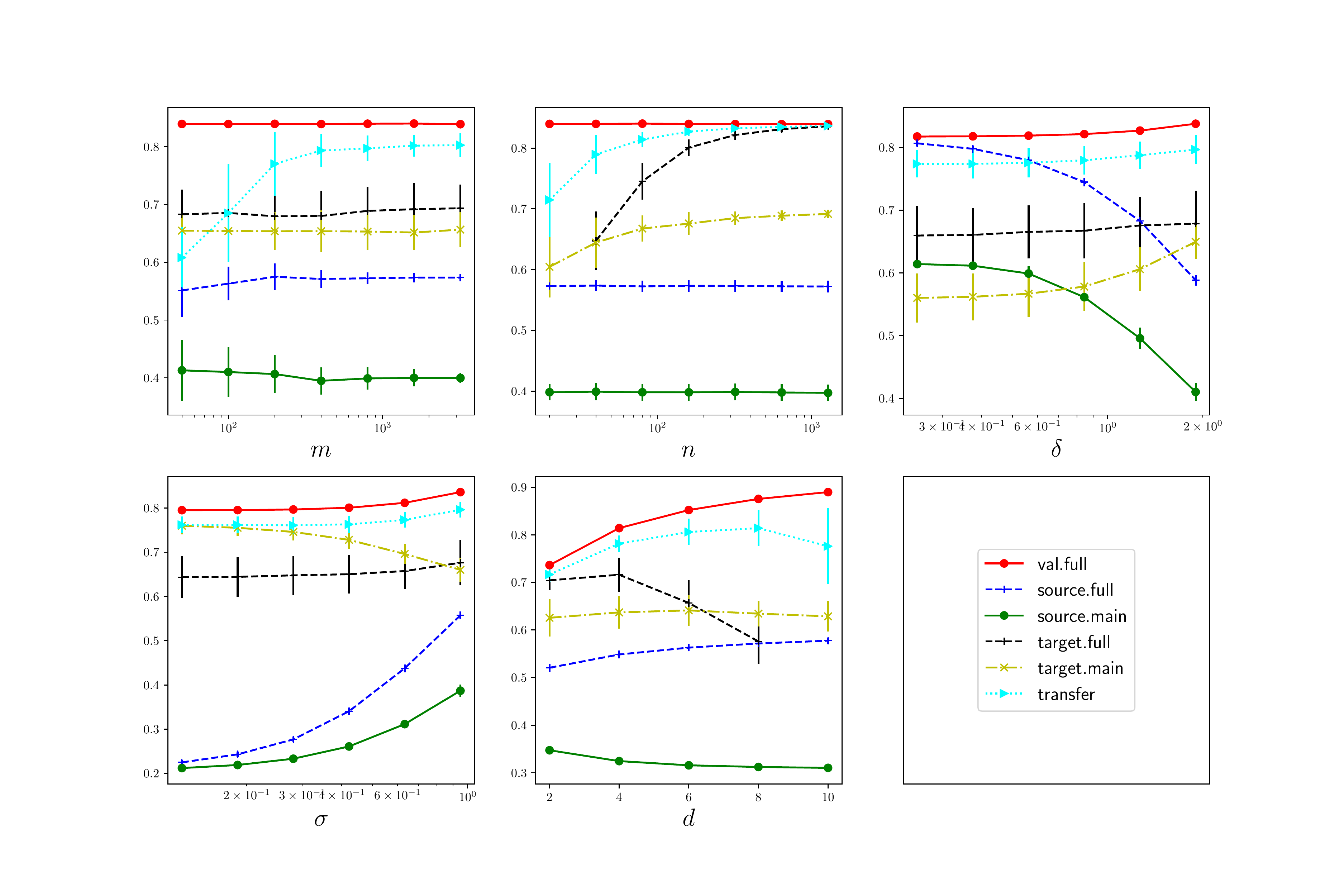}
    \caption{Predictive accuracy on target $Q$ for different source sample size ($m$, \emph{upper left}), target sample size ($n$, \emph{upper middle}), linear separation between $P$ and $Q$ ($\delta$, \emph{upper right}), non-linearity strength ($\xi$, \emph{lower left}) and covariate dimension ($d$, \emph{lower middle}). The \emph{cloglog} link is used in data generative process. The default values for this parameters are set at $m = 2000, n = 100, \delta = 2, \xi = 1$ and $d = 5$. }
    \label{fig:sim-logit}
\end{figure}

\newpage

\subsection{Mortality prediction with UK Biobank data}
\label{sec:appx-ukbb}

We present a list of predictors in the Table \ref{tab:ukbb-covariates}. 

    \begin{longtable}{@{\stepcounter{rowcount}\therowcount \hspace*{\tabcolsep}}  p{.60\textwidth}  p{.08\textwidth} p{.25\textwidth}} 
\toprule
Description & UDI & Type\\
\hline
 Sex; Uses data-coding 9 & 31 & Categorical (single)\\
 Waist circumference & 48 & Continuous\\
 Townsend deprivation index at recruitment & 189 & Continuous\\
 Number of days/week of moderate physical activity 10+ minutes; Uses data-coding 100291 & 884 & Integer\\
Number of days/week of vigorous physical activity 10+ minutes; Uses data-coding 100291 & 904 & Integer\\
Cheese intake; Uses data-coding 100377 & 1408 & Categorical (single)\\
 Milk type used; Uses data-coding 100387 & 1418 & Categorical (single)\\
 Cereal intake; Uses data-coding 100373 & 1458 & Integer\\
 Salt added to food; Uses data-coding 100394 & 1478 & Categorical (single)\\
Alcohol intake frequency.; Uses data-coding 100402 & 1558 & Categorical (single)\\
Cancer diagnosed by doctor; Uses data-coding 100603 & 2453 & Categorical (single)\\
Forced expiratory volume in 1-second (FEV1) & 3063 & Continuous\\
Diastolic blood pressure, automated reading & 4079 & Integer\\
Systolic blood pressure, automated reading & 4080 & Integer\\
Doctor restricts physical activity due to heart condition; Uses data-coding 100267 & 6014 & Categorical (single)\\
Qualifications; Uses data-coding 100305 & 6138 & Categorical (multiple)\\
Medication for cholesterol, blood pressure, diabetes, or take exogenous hormones; Uses data-coding 100626 & 6153 & Categorical (multiple)\\
Vitamin and mineral supplements; Uses data-coding 100629 & 6155 & Categorical (multiple)\\
Non-cancer illness code, self-reported; Uses data-coding 6 & 20002 & Categorical (multiple)\\
Treatment/medication code; Uses data-coding 4 & 20003 & Categorical (multiple)\\
Illnesses of father; Uses data-coding 1010 & 20107 & Categorical (single)\\
Illnesses of mother; Uses data-coding 1010 & 20110 & Categorical (single)\\
Illnesses of siblings; Uses data-coding 1010 & 20111 & Categorical (single)\\
Ethnic background; Uses data-coding 1001 & 21000 & Categorical (single)\\
Diagnosed with coeliac disease or gluten sensitivity; Uses data-coding 502 & 21068 & Categorical (single)\\
Weight & 23098 & Continuous\\
Body fat percentage & 23099 & Continuous\\
Body mass index (BMI) & 23104 & Continuous\\
Particulate matter air pollution (pm10); 2007 & 24019 & Continuous\\
Diagnoses - main ICD10; Uses data-coding 19 & 41202 & Categorical (multiple)\\
Beef intake; Uses data-coding 100016 & 103020 & Categorical (single)\\
Pork intake; Uses data-coding 100016 & 103030 & Categorical (single)\\
Fish consumer; Uses data-coding 100010 & 103140 & Categorical (single)\\
Vegetable consumers; Uses data-coding 100010 & 103990 & Categorical (single)\\
\bottomrule
    \caption{Descriptions, UDI codes and data types on the predictors in UK Biobank data}
    \label{tab:ukbb-covariates}
\end{longtable}

\subsubsection{Mortality prediction with logistic regression}

In addition to the models in Section \ref{sec:ukbb} we also studied the performance of three other models only with main effects. A description of them follows. 
\begin{itemize}

    \item \emph{source.main} fits only the main effects on White subpopulation: 
    \begin{equation}
         \begin{aligned}
    \hat \eta(x) \triangleq \sigmoid \left(\sum_{j}\hat \alpha_j x_j\right),\\
    \big(\{\hat \alpha_j\}_j\big) = \argmin \frac1m \sum_{i = 1}^m w_i^{(P)}\ell \left(y^{(P)}_i, \sum_{j} \alpha_j x_j^{(P)} \right) + \lambda \left(\sum_{j} |\alpha_j|\right)\,.
    \end{aligned}
     \label{eq:white.main-ukbb}
    \end{equation}

    \item And a \emph{target.main} fits only the main effects on Asian subpopulation: 
    \begin{equation}
      \begin{aligned}
    \hat \eta(x) \triangleq \sigmoid \left(\sum_{j}\hat \alpha_j x_j\right),\\
    \big(\{\hat \alpha_j\}_j\big) = \argmin \frac1n \sum_{i = 1}^n w_i^{(Q)} \ell \left(y^{(Q)}_i, \sum_{j} \alpha_j x_j^{(Q)} \right) + \lambda \left(\sum_{j} |\alpha_j|\right)\,.
    \end{aligned}  
    \label{eq:asian.main-ukbb}
    \end{equation}
    
    \item We can transfer the estimate of only the \emph{main effects} from the model fitted on White data and refit the linear part from the Asian data under a lasso penalty. Due to  penalization this is not the same as just estimating main effects from target data with lasso penalization, as done in target.main. We denote this model as \emph{transfer.main}. 
    \begin{equation}
        \begin{aligned}
    \hat \eta (x) \triangleq \sigmoid \left(\sum_{j}\hat \alpha_j^{(Q)} x_j \right),\\
    \big(\{\hat \alpha_j^{(P)}\}_j\big) = \argmin \frac1m \sum_{i = 1}^m w_i^{(P)}\ell \left(y^{(P)}_i, \sum_{j} \alpha_j x_j^{(P)} \right)
    + \lambda \left(\sum_{j} |\alpha_j| \right),\\
  \big(\{\hat \delta_j\}_j\big) = \argmin \frac1n \sum_{i = 1}^n w_i^{(Q)}\ell \left(y^{(Q)}_i, \sum_{j} (\hat \alpha_j^{(P)} + \delta_j) x_j^{(P)} \right)
  + \lambda \left(\sum_{j} |\delta_j| \right),\\
  \hat \alpha_j^{(Q)} = \hat \alpha_j^{(P)} + \hat \delta_j
    \,.
    \end{aligned}
    \label{eq:offset-main}
    \end{equation}

\end{itemize}

With these there extra models we present the full results in Table \ref{tab:ukbb-full}.

\begin{table}[]
    \centering
   \begin{tabular}{lllrrrrr}
\toprule
{} &   train-data & test-data &  accuracy &    TPR &    TNR &    auc &  bal-acc \\
\midrule
source.full &        White &     White &     0.746 &  0.745 &  0.746 &  0.814 &    0.746 \\
            &        White &     Asian &     0.773 &  0.680 &  0.775 &  0.769 &    0.728 \\
\hline
source.main &        White &     White &     0.746 &  0.745 &  0.746 &  0.813 &    0.746 \\
            &        White &     Asian &     0.772 &  0.680 &  0.775 &  0.765 &    0.727 \\
\hline
target.full &        Asian &     Asian &     0.759 &  0.640 &  0.762 &  0.751 &    0.701 \\
target.main &        Asian &     Asian &     0.740 &  0.620 &  0.743 &  0.744 &    0.681 \\
\hline
transfer.main        &  White/Asian &     Asian &     0.753 &  0.700 &  0.755 &  0.765 &    0.727 \\
transfer &  White/Asian &     Asian &     0.750 &  0.740 &  0.750 &  0.773 &    0.745 \\
\bottomrule
\end{tabular}
\caption{Accuracies of 10 year mortality prediction in White and Asian subpopulations under different models.  }
    \label{tab:ukbb-full}
\end{table}

\end{document}